\documentclass[journal]{IEEEtran}

\usepackage{times}
\usepackage{amsmath,amsfonts,amssymb,color}
\usepackage{graphicx,subfigure}
\usepackage[mathscr]{eucal}
\newenvironment{proof}{\noindent{\em Proof :}}{ \hfill $\Box$ \\}
\usepackage{cite}
\usepackage{array}
\usepackage{float}
\usepackage{epsf}
\usepackage{fancyhdr}
\usepackage{subfigure}
\usepackage{latexsym}
\usepackage{fixltx2e}

\newtheorem{theorem}{Theorem}[section]
\newtheorem{definition}{Definition}[section]
\newtheorem{proposition}{Proposition}[section]
\newtheorem{lemma}{Lemma}[section]
\newtheorem{corollary}{Corollary}[section]
\newtheorem{example}{Example}[section]

\newcommand{\q}{\quad}
\newcommand{\qq}{\qquad}
\newcommand{\Tr} {\mathrm{tr}}
\newcommand{\var} {\mathrm{var}}
\newcommand{\diag} {\mathrm{diag}}
\newcommand{\rank} {\mathrm{rank}\,}
\newcommand{\Var} {\mathrm{Var}\,}

\newcommand{\Frac}[2]{{\displaystyle\frac{#1}{#2}}}

\newcommand{\Span} {\mathrm{span}\,}
\newcommand{\virg}[1]{``#1''}

\newcommand{\bmat}{\left[ \begin{matrix}}
\newcommand{\emat}{\end{matrix} \right]}

\newcommand{\E}{{\mathbb E}\,}

\newcommand{\Rbb}{\mathbb R}

\newcommand{\Zbb}{\mathbb Z}

\newcommand{\xb}{\mathbf  x}
\newcommand{\yb}{\mathbf  y}
\newcommand{\zb}{\mathbf  z}
\newcommand{\wb}{\mathbf  w}
\newcommand{\vb}{\mathbf  v}
\newcommand{\eb}{\mathbf  e}

\newcommand{\ub}{\mathbf  u}

\newcommand{\Hb}{\mathbf H}
\newcommand{\Ib}{\mathbf I}

\def\xib{\boldsymbol{\xi}}
\def\etab{\boldsymbol{\eta}}

\title{Modeling complex systems by \\ Generalized Factor Analysis }

\author{ Giulio Bottegal and Giorgio  Picci~\IEEEmembership{Life~Fellow,~IEEE}
\thanks{G. Bottegal is with  the ACCESS Linnaeus Centre, School of Electrical Engineering,
KTH Royal Institute of Technology, SE-100 44 Stockholm, Sweden;   {\tt\small bottegal@kth.se}}
\thanks{ G. Picci is with   the Department of Information
Engineering, University of Padova, Padova, Italy;  {\tt\small picci@dei.unipd.it}} }

\begin{document}

\maketitle

\begin{abstract}
We  propose a new modeling paradigm for large dimensional aggregates of stochastic systems by  Generalized Factor Analysis (GFA)  models. These models describe the data as the sum of a {\em flocking} plus an uncorrelated {\em idiosyncratic } component. The flocking component describes a sort of collective  orderly motion  which   admits a much  simpler   mathematical description than the whole ensemble while the idiosyncratic component describes weakly correlated noise. We first discuss  static GFA representations and characterize  in a rigorous way the properties of the two components. The extraction of the dynamic flocking component  is  discussed  for time-stationary linear systems and  for a simple classes of separable  random fields.
\end{abstract}
\thispagestyle{empty}

\section{Introduction}\label{Introd}
It has been observed in several circumstances \cite{Deistler-Z-07,anderson_deistler_2008,Deistler-A-F-Z-010,Deistler-etal-2010} that modeling and identification of complex stochastic systems by traditional AR or ARMA models may lead to  problems where the number of parameters can be  of the same order of magnitude or larger than the sample size. The only way out of this problem seems to be to change our ideas on modeling. In this paper we propose a new paradigm on stochastic modeling of complex systems based on the theory of {\em Generalized Factor Analysis (GFA)} and the idea of {\em stochastic flocking}. Although the two terminologies belong  to different cultures which seem to have little in common, our point in this paper will be to show that   dynamic GFA modeling of a large  ensemble of interacting random units   hinges  on    splitting   the overall motion   into a  component which deserves the name of {\em flocking}  plus a  weakly correlated kind of noise. The latter  is called the  {\em idiosyncratic} component. The first component describes the average random motion of the system by a rather simple statistical model while the second aims at describing the stochastic dynamics which pertains exclusively to   individual fluctuations about the average.

 The word   {\it Flocking} is used to describe     a commonly observed  behavior in gregarious animals  by which many equal individuals tend to group and follow, at least approximately, a common path in space. The  phenomenon  has been studied very actively in recent years; see e.g.  \cite{Reynolds-87,Vicsek-etal-95,Veerman-etal-05,Brockett_10} and the literature  on this subject  is now huge, consisting of  hundreds of papers which would be  impossible to discuss here. Our interest in flocking derives from the fact that the phenomenon has similarities with many scenarios observed in artificial/technological   environments a  few examples of which will be described below.

 The mechanism of  formation  of flocks is sometimes also called {\em convergence to consensus}  and has   been intensely studied in the literature. There is now a quite articulated theory   addressing the convergence to consensus under  a variety of assumptions on the communication strategy among agents, specific nonlinearities of the dynamics, the kind of permissible local control actions etc. see e.g.  \cite{Jadbabaie-L-M-03, Fagnani-Z-08,Olfati-F-M-07,Tahbaz-J-10,Cucker-S-07,Olfati_06, Shen-07,Tanner-J-P-07} and   references therein.

 In this paper we want to address a different and possibly more basic issue: given observations of the motion of a large set of interacting agents and assuming statistical steady state, find out  whether there is a flocking component in the collective motion and estimate its   characteristics. The rationale for   this search is that the very concept of flocking implies an  {\em orderly motion} which must then admit a much simpler   mathematical description than that of the whole ensemble. Once the     flocking component (if present) has been separated, the motion of the ensemble   splits naturally into flocking plus  a random term which describes local random disagreements  of the individual agents or the effect of external disturbances. Hence extracting a flocking structure is essentially  a parsimonious modeling problem. Prediction of the future behavior and control of a complex  ensemble of random agents could then reasonably  be restricted to the flocking component and be based on the simple model thereof.
  \subsection{Problem statement and scope of the paper}\label{ProblemSect}
We start by setting  notations:  In this paper boldface symbols will normally denote
random arrays, either finite or infinite. All random variables will be real,  zero-mean and with finite variance. In
the following we shall denote by the symbol $H(\vb)$ the standard
Hilbert space of random variables linearly generated by the
scalar components $\{\vb_1,\ldots,\vb_n,\ldots \}$ of a (possibly infinite)
family of random variables which we generically denote  $\vb$. For $\xib,\,\etab \in H(\vb)$, the inner product  is the mathematical expectation $\langle \xib,\,\etab\rangle:= \E \xib \etab$ which induces  the (variance) norm of  random variables by settig $\|\xib\|^2 = \E \xib^2$.   Convergence of random sequences will always be understood with  respect to this norm.

Let $\yb(k,t)$ be a   finite variance random field depending on a discrete space variable $k$ and on a time variable $t$. We  shall denote by $\yb(t)$ the random (column) vector with components $\{\yb(k,t)\,;\, k=1,2,\ldots,N\}$. Suitable mathematical assumptions on this process will be  specified in due time. The variable $k$ is indexing  (space) locations of   a large ensemble of  ``agents''  each of which produces at   time $t$  a scalar measurement, $\yb(k,t)$, of an observable quantity \footnote{An extension of the theory presented in this paper  to a more general setting where each component  $\yb(k,t)$  may  take  vector values, say in $\Rbb^m$,  is  conceptually straightforward, although of course at the price of more complicated notations. For the sake of clarity we shall here restrict to scalar-valued processes.}.
We shall assume that $k$ varies on some ordered  index set of $N$ elements and let $t \in \Zbb$ or $\Zbb_{+}$, depending on the context. Eventually we shall be interested in problems where $N=\infty$. The following is a first attempt to define in precise terms a random flock. The definition is given for a finite ensemble and, as it stands,   may lead to non-unique modeling of the same system, which is instead of paramount importance    in statistical identification  theory. The notion will have to be refined later for an infinite ensemble. See Sections \ref{sect:Aggreg} and \ref{Sec:RandFields}.

A $q$-dimensional {\bf  random  flock} is a random field having the multiplicative  structure  $\hat{\yb}(k,t)= \sum_{i=1}^{q} \,f_i(k) \xb_{i}(t)$, or equivalently,
\begin{equation}
\hat{\yb}(t)=   \sum_{i=1}^{q} \,f_i  \xb_{i}(t)
\end{equation}
where $f_i= \begin{bmatrix} f_i(1) &  f_i(2)& \ldots & f_i(N)\end{bmatrix}^{\top},\, i=1,2,\ldots,q$ are nonrandom $N$-vectors which do not depend on time and $\xb(t):= \begin{bmatrix} \xb_{1}(t) & \ldots&\xb_{q}(t)\end{bmatrix}^{\top}$ is a random processes  depending on the time variable only, which can be chosen with orthonormal components; i.e.
$ \E \xb(t)\xb(t)^{\top}= I_q \,,\q t \in \Zbb\, $.\\
The idea is that a random flock can  essentially be regarded as a rigid  deterministic geometric configuration of $N$ objects (or points) in  space moving synchronously in a random fashion. A very simple  intuitive  picture can be imagined   extending for a moment our model to allow for three dimensional (vector valued)  outputs  $\yb(k,t)\in \Rbb^3\,, k=1,2,\ldots $. The $k-$th agent  can then be visualized as a point  moving in 3-D space. Let $q$ be also equal to three and think of the $3-$dimensional random motion with coordinates  $\xb(t)$ as the  motion of, say, the barycenter of the  ensemble. This implies that all different agents follow the same trajectory in $3$-dimensional space, modulo a constant offset depending on their relative location. In general however the agent's output  does not need to be of the same dimension  of the common state $\xb(t)$.  As already said, here for simplicity we restrict to one dimensional output components.

The above may look like   a rather crude mathematical idealization of  animal flocking behavior especially in that the spatial pattern of flocks of birds or herds of animals etc.  may  also deform its shape in time \cite{Hemelrijk-H-011}. Note however that the deformations may be interpreted as random fluctuations   about an average common trajectory that all agents of the flock aim at and that these fluctuations could in principle  be embodied in the ``noisy component'' $\tilde{\yb}$ of our  modeling scheme.  The denomination of  random  flock  above may be reasonable as a  description of the average behavior of a realistic  flock.
 It should however be said very clearly that our objective in this paper is not to address animal behavior, but rather to discuss dynamical modeling of complex technological systems.

The main goal of  this paper is to  investigate when a second order random field has a flocking component and study the problem of extracting it  from sample measurements of  $\yb(k,t)$.  This means  that one should be searching for decompositions of the type:
\begin{equation}\label{Repr:GFA}
\yb(t)=   \sum_{i=1}^{q} \,f_i  \xb_{i}(t) +\tilde{\yb}(t)
\end{equation}
where $q \geq 1$ and $\tilde{\yb}(t)$ is a ``random noise'' field which should not contain flocking components.
Naturally for the problem to be well-defined one has to specify conditions making this decomposition {\it unique}.

\subsection{Examples}

\subsubsection{Detection of emitters}

In this scenario we suppose there is an unknown number, say $q$, of
emitters, each of them broadcasting radio impulse trains at a fixed common
frequency. Such impulses are received by a large array of $N$ antennas
spread in space. The measurement of each antenna is corrupted by noise,
generated by measurement errors or local disturbances, possibly correlated with that of neighboring antennas. The set up can be described mathematically, by indexing each antenna by an
integer $i=1,2,\ldots, N$  and   denoting by $\yb_i(t)$ the signal
received at time $t$ by antenna $i$. Then, model \eqref{Repr:GFA} can be
used to describe the received signal, with  $\xb (t)$ being the signals sent by the emitters at time $t$, $f_{i}$ coefficients related to the distance between the emitters and antenna $i$ and $\tilde \yb_i(t)$ some disturbance affecting antenna $i$ at time $t$. Hence, we may identify $\sum_{i=1}^{q} \,f_i  \xb_{i}(t)  $ as the flocking component of $\yb(t)$. The goal is to detect the number of emitters $q$ and possibly estimate the signal components impinging on the antenna array. 
Note that in the model there are several hidden (non-measurable) variables, including  the dimension $q$. In our setting   $N$ is assumed to be very large;  ideally we shall assume  $N \rightarrow \infty$.
%
One should  note that estimation of this model from    observations $\{y(t)\}$ of $\yb(t)$ consists first of estimating the model parameters, say $\{f_i\}$ and the covariance matrix of $\tilde \yb(t)$ but also in constructing the hidden random quantities  $\xb(t)$ and $\tilde \yb(t)$. The covariance matrix of $\yb(t)$, say $\Sigma \in \Rbb^{N\times N}$ may be obtained from the data by standard procedures.

 A  problem leading to models of similar structure is automated speaker detection. This  is the problem of detecting the speaking persons (emitters)  in a noisy environment at any particular time, from  signals coming from  a large array of $N$ microphones distributed in a room. Here the number of emitters is generally small but could be varying with time. Robustly solving this problem is useful in areas such as surveillance systems,  and human-machine interaction.

%

\subsubsection{Inference on gene regulatory networks}

 \emph{Transcription factors (TFs)} are proteins which regulate gene
expression binding to specific sequences of the promoter region of a
gene. This regulation brings to the transcription of genes into
mRNAs, which are in turn translated into proteins giving rise to a
complex gene regulatory network.  
In a schematic representation of this network, TF regulation of genes is represented by directed links with a weight
proportional to the strength of the regulation on each gene; possible mutual interactions among genes are also accounted for \cite{pournara_2007,fa_chap_2010}. 
Usually, measurements of the activity of a large number of genes can be collected, but no
information is available about their regulators (i.e., the TFs). Hence, retrieving the TF activity from gene expression data is a challenging problem in gene regulatory networks.

Referring to the model \eqref{Repr:GFA}, we may use the vector $\yb(t)$ to represent the measured
expression profile of the genes of the network in the $t$-th experiment. Describe the TF activity by the variable~$\xb(t)$, the strength of the effect of the TFs on the $i$-th gene of the network by the loading vector $f_i$'s  and the gene mutual influences by the random vector  $ \tilde \yb(t)$,  we   obtain a description of the network as a \virg{flocking component plus local interactions} model. Estimating the flocking component due  to the action $\xb(t)$ of the TF's may constitute a preliminary step towards  understanding  the TF activity in the network.
%

\subsubsection{Modeling energy consumption}

In this example, we may want to model the energy consumption (or production) of a network of $N$ users  distributed geographically in a certain area, say a city or a region. The energy consumption $\yb_i(t)$ of user $i$  is a random
variable which can be seen as  the sum of the two contributions in  \eqref{Repr:GFA},
where the term $f_i^\top  \xb(t)$, the   flocking component of the model,
represents a linear combination of $q$  hidden variables $\xb_i(t)$ which model different factors affecting the energy consumption (or production) of the whole ensemble; say heating or air conditioning consumption related to seasonal climatic  variations, energy production related to the current status of the economy etc. The factor vector $\xb(t)$    determines the average time pattern of energy consumption/production  of each unit, the importance of each scalar factor being determined by    a $q$-ple of constant weight coefficients $f_{i,k}$. The terms $\tilde \yb_i(t)$, represent local  random fluctuations which model the consumption due to  appliances or devices that are usually activated randomly, for short periods of time. They   are assumed  uncorrelated with the process $\xb$. The covariance  $\E \tilde \yb_i(t) \tilde \yb_j(t)$ could  be non-zero for neighboring users but is reasonable to expect that  it decays to  zero when $|i-j|$ is large. \\
To identify such a model one should start from real data  of energy consumption collected from a large amount of units.  A possible application for such a model is the forecasting of the average requirement of energy in a certain geographical area.

\subsubsection{Dynamic modeling in computer vision}

Large-dimensional time series occur often in signal processing applications,   typically  for example,    in computer vision and dynamic image processing. The role of identification  in image processing and computer vision  has been addressed by several authors.  We may refer the reader to the    survey \cite{Chiuso-P-08} for more details and references. One starts  from  a signal $\yb(t) := {\rm vec}(\Ib(\cdot,t))$, obtained by vectorizing at  each time $t$, the    intensities $\Ib(\cdot,t)$ at each pixel of an image, into a vector, say $\yb(t) \in \mathbb{R}^N$, with a ``large'' number (typically tens of thousands) of components.      We may for instance be  interested in modeling (and in identification
methodologies thereof) of
``dynamic textures'' (see   \cite{DynamicTexture}), by  linear state space models or  in extracting classes of models describing   rigid motions of objects of a scene. Most of these models involve   hidden variables, say the state of linear models of textures, or the displacement-angular velocity coordinates of  the rigid motions of objects in the scene. The purpose is of course  to compress high dimensional data into simple mathematical structures. Note
that the number  of samples that can be used for identification  is  very often of
the same order (and sometimes smaller) than the data dimensionality. For instance, in dynamic textures modeling, the number of images in the sequences is of the order of a few hundreds while
$N$ (which is equal to the number of pixels of the image) is
certainly of the order of a few hundreds or thousands \cite{DynamicTexture,BissaccoCSPAMI2007}.

\subsection{Structure of the paper}

The organization of the paper is as follows:
In Section \ref{Sec:FA}  we quickly review static finite-dimensional Factor Analysis; in Section \ref{Sec:AggrIdio} following  the  basic definition  of  of \cite{forni_lippi_2001},   we provide a neat  mathematical characterization of idiosyncratic sequences (Theorem 3.1) which is believed to be new. In the following Section \ref{sect:Aggreg}, based on the characterization of idiosyncratic sequences, the  notion of strong linear independence is introduced and shown to be a crucial ingredient to provide    a unique  representations  by GFA models.  New results characterizing the factor loadings  are presented in this section and related to the condition of diverging eigenvalues from the literature on GFA.  The problem of representation by GFA  models is discussed in Section \ref{GFA}. Here the proposed notion of $q$-aggregate sequence from \cite{forni_lippi_2001} is made concretely operational and a procedure to compute asymptotically the factor and the idiosyncratic components is presented for the first time.  The literature on GFA does not seem to distinguish between GFA representations of a {\em covariance matrix} and of  a {\em random sequence}. However while the first may be unique  there may be quite diverse random components  $\hat{\yb}, \tilde{\yb}$ yielding a GFA representation of the same string $\yb$. To guarantee uniqueness one needs for example to impose that  $\hat{\yb}, \tilde{\yb}$ have components in $H(\yb)$. An interpretations of the two GFA components  in terms of short and long range interaction of a large ensemble of stochastic agents is discussed in Section \ref{SubsectA}.  Also, new necessary and sufficient conditions for  a (weakly) stationary sequence to have a GFA representation are presented at the end of the section in \ref{Sec:Wold}. In Section \ref{Sec:RandFields}, which is believed to be completely original,  dynamic GFA representations of two classes of random fields of interest in applications are discussed. Time-stationary  random fields provide in particular a wide class of linear stochastic models which can describe complex systems arising in a variety of applications in the system and control area.  The extraction of the flocking  component for these systems reduces to the study of an infinite dimensional Lyapunov equation.

Some of the material of this paper has been presented in a preliminary form at  conferences \cite{Bottegal-P-11,Bottegal-P-13}.

\section{A short review of static Factor Analysis models}\label{Sec:FA}

Factor Analysis (FA) has a long history; it has apparently first been introduced by psychologists
and     successively been studied and applied in various branches of Statistics and Econometrics \cite{Ledermann-37,Ledermann-39,Bekker-L-87,Lawley-M-71}. Dynamic versions of factor models  have also  been introduced  in the  econometric literature, see e.g.
\cite{geweke_1977, PenaBox1987,PenaBoxModel2004} and
references therein.
With a few  exceptions however,
\cite{KalmanDFM1983,Kalman-83,Picci-1987,picci_pinzoni_86,Georgiou-N-011,Lipeng-etal-2013},
little attention has been payed to these models in the system and
control engineering community.
 Recently, we have been  witnessing a revival of interest in Factor Analysis, due to the generalization
  proposed  by Chamberlain, Rothschild, Forni, Lippi and collaborators in a series of widely quoted papers \cite{chamberlain_1983,Chamberlain-R-83,Forni_2000,forni_lippi_2001}. This new modeling paradigm is
attracting a considerable attention also in the   system identification community \cite{anderson_deistler_2008,Deistler-Z-07,Deistler-A-F-Z-010,Deistler-etal-2010,PenaJSPI2006}. The   new models, called \emph{Generalized   Factor  Analysis   (GFA)} models, although initially motivated by  financial econometrics   seem to  have a potential to be  useful also in engineering applications.

A classical (static) {\it Factor Analysis\/} model is a representation of $N$ observable random variables $\yb=[\,\yb(1)\,\ldots\,\yb(N)\,]^{\top}$, as linear combinations
of $q$ {\it common factors\/} $\xb=[\,\xb_1\,\ldots\,\xb_q\,]^{\top}$,
plus uncorrelated ``noise" or ``error" terms $\eb=[\,\eb(1)\,\ldots\,\eb(N)\,]^{\top}$ of the type
\begin{equation}
\yb=F\xb+\eb, \label{FA}
\end{equation}
 The columns $\{f_1,\;f_2,\;\ldots,f_{q}\}$ of the matrix
$F$, called the {\it factor loadings} can be chosen linearly independent    and the common factors can be  normalized in such a way  that $  \E \xb\xb^{\top} = I$, which we shall always assume in the following.    An essential part of the model
specification is  that the  $N$ components of the error $\eb$ should
be (zero-mean and) mutually uncorrelated random variables, i.e.
\begin{equation}
 \E  \xb\eb^{\top}=0\,,  \qquad \E
\eb\eb^{\top}=\diag\{\sigma^{2}_1,\ldots,\sigma^{2}_N\}\,.
\end{equation}
The aim of these models is to provide an ``explanation" of the
mutual correlations of the observable variables $\yb(i)$ in terms
of a small number of common factors, in the sense that, setting: $ \hat{\yb}(k):=\sum f_i(k)\xb_i$, where $f_{i}(k)$ is the $k$-th component of $f_i$, one has exactly   $\E \yb(i)\yb(j)=\E \hat{\yb}(i)\hat{\yb}(j)$,
for all $i\neq j$. Note that  a FA representation then induces a decomposition of the covariance matrix $\Sigma$ of $\yb$ as
\begin{equation}
\Sigma= F F^{\top} + \diag\{\sigma_{\eb_1}^2,\ldots,\sigma_{\eb_N}^2\}
\end{equation}
which can be seen as a special kind of \emph{low rank  plus sparse} decomposition of a covariance matrix \cite{Chandra-etal-011}, a diagonal matrix  being, in intuitive terms, as sparse as one could possibly ask for.

Unfortunately these models, although providing in many circumstances a quite natural and
useful data compression scheme,  suffer from a
serious non-uniqueness problem coming from the fact that, even for a fixed dimension $q$ there are in
general many (generally infinitely many) statistically non-equivalent FA models describing the same  family of observables $\{\yb(1),\ldots,\yb(N)\}$. In addition, determining the minimal integer $q$ for which a FA decomposition holds for a given symmetric positive definite matrix $\Sigma$ has been an open problem since the beginning of the last century. It is by
now a well-known fact that estimation of  F.A. models (say minimal models with $F$'s of  rank $q$ and   normalized factors) is an ill-posed problem.
 
  This inherent nonuniqueness    is  called ``factor indeterminacy", or
unindentifiability in the literature and the  often acritical use of FA models has been vehemently criticized by Kalman in a series of papers, see e.g.  \cite{Kalman-83,KalmanDFM1983}. Non uniqueness of the factors is an intrinsic difficulty common to stochastic models with latent variables, whose role is to enforce some    conditional independence relation among the observables.  As a rule the   choice of the latent variables is non unique. 
It is known, see  \cite{Bartholomew-84,Picci-1987,LPBook},  that a factor structure is also  equivalent to  a relation of  conditional independence of the observables given the factors and this  is in turn {\em equivalent to the uncorrelation of the noise components.} \\ 
 One may then try to obtain uniqueness by giving up or by mitigating the requirement of uncorrelation of the components of  $\eb$. Obviously this tends to make the problem  ill-defined  as the basic goal  of uniquely splitting the external observable signal into a noiseless component plus ``additive noise'' is made vacuous, unless some  extra assumptions are made on the model and on the very notion of ``noise''.   Quite surprisingly, as we shall see, for  models describing an {\em infinite } number of observables a meaningful weakening of the uncorrelation property can be introduced, so as to guarantee  the uniqueness of the decomposition.

\section{ Static Generalized Factor Analysis and Idiosyncratic sequences} \label{Sec:AggrIdio}
In this section we shall review  Generalized Factor Analysis    restricting  for now to the {\em static case}.\\
 Consider a  zero-mean   finite
variance   stochastic process $\yb :=\{\yb(k),\,k \in \Zbb_+\}$ represented  as a  column vector with an infinite number of random components.
We  want   to represent $\yb$  by an infinite dimensional Factor Analysis model o fthe form
\begin{equation}\label{GFAD}
\yb(k) = \sum_{i=1}^{q}f_i(k) \xb_i + \tilde{\yb}(k)\,,\qq k=1,2, \ldots
\end{equation}
where, in analogy to  finite-dimensional Factor Analysis,  the random variables $\xb_i\,,\, i=1,\ldots,q$ are the  common factors  and the deterministic vectors $f_i\in \Rbb^{\infty}$ the   factor loadings. The $\xb_i$ form a q-dimensional  random vector  $\xb$  with orthonormal components; i.e.  $\E\;\xb\xb^{\top} =I_q$.   The  $\tilde{\yb}(k)$'s are zero mean  random variables   orthogonal to (uncorrelated with) $\xb$. The linear combinations $ \hat{\yb}(k):=\sum f_i(k)\xb_i\,; \;k=1,2,\ldots$  are the components of  an  infinite random vector denoted $\hat{\yb}$ which, together with  the noise terms $\tilde{\yb}(k)$, give the representation  $\yb =\hat {\yb} + \tilde{\yb}$, a compact way to write the model \eqref{GFAD} in vector form.\\
 Which specific characteristics  qualify the process $\tilde{\yb}$    as ``noise''  is a nontrivial issue  which will be  the main theme  of this section and will lead to the concept  of {\em  idiosyncratic random sequence} below.  The underlying idea can be extracted from the following example.
 \begin{example}\label{Stringof1}
Let ${\mathbf 1}\!\!1$ be an infinite column  vector of $1$'s,   let $\xb $ be a zero-mean scalar random  variable and $ \tilde{\yb}$  a zero-mean  weakly stationary    ergodic sequence uncorrelated with $\xb $. Consider the F.A. model
\begin{equation*}
\yb = {\mathbf 1}\!\!1 \xb + \tilde{\yb}\,
\end{equation*}
and the   sequence of vectors in $\Rbb^{\infty}$
\begin{equation}\label{Aritmean}
a_n = \frac{1}{n}  [\,\underbrace{1 \, \ldots
\,1}_{n}\,0\,\ldots\, ]^{\top}
\end{equation}
 Since $\lim_{n \to \infty}\,a_n^{\top}{\mathbf 1}\!\!1=1$ and  $\lim_{n \to \infty}\,a_n^{\top}\tilde{\yb}= \lim_{n \to \infty}\, \dfrac{1}{n} \sum_{k=1}^n \tilde{\yb}(k) = \E  \tilde{\yb}(k)= 0$ (limit in $L^2$) we have
$
 \lim_{n \to \infty}\,a_n^{\top} \yb  = \xb\,;
$
 hence  we can recover the latent factor  by an operation of averaging. There are in fact more general sequences $a_n$ of infinite   vectors, such that $\lim_{n\rightarrow \infty} a_n^{\top} {\mathbf 1}\!\!1$ exists and is non zero and in addition   $\lim_{n\rightarrow \infty} a_n^{\top} \tilde{\yb} = 0$ for processes like $ \tilde{\yb}$.  These sequences  recover $\xb$ from the FA model uniquely. \hfill $\Box$
\end{example}

The infinite covariance matrix of the vector $\yb$ is formally written as  $\Sigma := \mathbb{E}\;  \yb \yb^{\top} $. Let $\ell^2(\Sigma)$ denote the Hilbert space of
infinite sequences $a:= \{a(k),\,k \in \Zbb_+\}$ such that $\|
a \|^2_\Sigma := a^{\top} \Sigma a < \infty$. When $\Sigma =I$ we use the standard symbol $\ell^2$, denoting the corresponding norm by $\|\cdot\|_2$. The following definition was introduced in \cite{forni_lippi_2001}:

\begin{definition} \label{Idiosync}
{\em A sequence of elements $\{a_n\}_{n \in \Zbb_+} \subset \ell^2 \cap \ell^2(\Sigma)$   is an {\em averaging sequence} (AS) for $\yb$, if
$\lim_{n\rightarrow\infty} \| a_n \|_2 = 0$. \\
  We say that a sequence of random variables $\yb$ is   {\bf idiosyncratic}   if
$\lim_{n\rightarrow\infty} a_n^{\top} \yb = 0$ for any averaging
sequence $a_n \in \ell^2 \cap \ell^2(\Sigma)$.}
\end{definition}
\medskip

Whenever  the covariance $\Sigma$ is a bounded operator on $\ell^2$ one has $\ell^2(\Sigma) \subset \ell^2$;  in this case an AS can be seen just as a sequence of linear functionals in $\ell^2 $ converging strongly to zero. \\
The sequence of elements \eqref{Aritmean} is in $\ell^2$ and 
is an averaging sequence for any $\Sigma$. For a more general class of AS's, let $P_n $ denote the compression of the $n$-th power of the left shift operator to the space $\ell^2$; i.e. $[P_n a] (k)=  a(k-n) $ for $k \geq n$ and zero otherwise. Then $\lim _{n \to \infty}\, P_n a=0$ for all $a\in \ell^2$ \cite{Halmos-61} so that $\{P_n a\}_{n \in \Zbb_+}$ is an AS for any $a \in \ell^2$.
\medskip

The nature of an idiosyncratic sequence  is   related
to certain properties  of  its covariance matrix. To explain  this point, we need to
  introduce some notations and facts about the eigenvalues of sequences of
covariance matrices. 
We let  $\Sigma_n$ indicate  the top-left $n \times n$ block of
$\Sigma$, equal to the covariance matrix of the first $n$
components of $\yb$, the corresponding $n$-dimensional vector being denoted by  $\yb^{n}$. The inequality $\Sigma > 0$  means that all submatrices $\Sigma_n$  of $\Sigma$  are   positive  definite, which we shall always assume in the following.
Letting  $\hat{\Sigma} := \mathbb{E}\hat{\yb} \hat{\yb} ^{\top}= FF^{\top}$ and
$\tilde{\Sigma}  := \mathbb{E}\tilde{\yb} \tilde{\yb} ^{\top}$, the orthogonality   of the noise term and the factor
components implies that
\begin{equation}\label{FADecomp}
\Sigma  = \hat{\Sigma}  + \tilde{\Sigma}\,,
\end{equation}
that is, $\Sigma_n = \hat{\Sigma}_n + \tilde{\Sigma}_n \,,\, 
\forall n \in \mathbb{Z}_+\,$. Even imposing that $\hat{\Sigma}$ should be of low
rank, this is a priori  a highly non-unique decomposition. There are
situations/examples  in which the $\tilde{\Sigma}$ is diagonal as in
the finite-dimensional Factor Analysis case, but these situations are
exceptional.
Denote by $\lambda_{k}(\Sigma_n)$  the k--th
eigenvalue of the $n \times n$ upper left submatrix $\Sigma_n$ of $\Sigma$. The $\lambda_{k}(\Sigma_n)$'s are real nonnegative and will always  be ordered by decreasing magnitude. By Weyl's theorem \cite[p. 203]{Stewart-S-90}, see also   \cite[Fact M]{forni_lippi_2001},     the k--th eigenvalue of $\Sigma_n$ is a non decreasing function of $n$ and hence has a limit, $\lambda_k(\Sigma)$, which may possibly be equal to  $+\infty$. Each such limit is called an {\em  eigenvalue of $\Sigma$}.
These limits however  are in general not true eigenvalues, as it is well-known that $\Sigma$ may not have eigenvalues. For example, a bounded symmetric Toeplitz matrix has a purely continuous spectrum \cite{Hartman-W-54}. Anyway since $\Sigma$ is symmetric and positive, its spectrum lies on the positive half line and its elements  can   be ordered. Henceforth   we  shall denote by $\lambda_1(\Sigma)$ the maximal eigenvalue of $\Sigma$, as  defined  above, with the convention that  $\lambda_1(\Sigma)=+\infty $ when there are infinite eigenvalues. The following Lemma will be instrumental in understanding  the structure of idiosyncratic processes.
\begin{lemma}\label{th:Boundedness}
{\em  A symmetric  matrix $\Sigma$ defines  a bounded operator on $\ell^2$ if and only if $\lambda_1(\Sigma)$ is finite.}
\end{lemma}
\begin{proof}
see the appendix.
\end{proof}
A   characterization of  idiosyncratic sequences is stated in
the following theorem.  The proof will also be given in the appendix.
 \begin{theorem} \label{th:idiosyncratic}
 {\em The sequence $\yb$ is idiosyncratic   if and only if its covariance matrix defines a bounded operator on $\ell^2$.}
\end{theorem}
In particular, a white noise  process with uniformly bounded variance is idiosyncratic. This follows since the covariance of a white noise process is a diagonal matrix with uniformly bounded entries and therefore is a bounded operator.   However the notion of idiosyncratic process is much  more general than that of white noise. For example any (weakly) stationary purely non deterministic process with a bounded spectral density is idiosyncratic. See Section \ref{Sec:Wold}.

A test for idiosyncracy of a random sequence can be based on    Lemma \ref{th:Boundedness} whereby  $\yb$ is idiosyncratic if and only if $\lambda_1(\Sigma)$ is finite; this last characterization is due to  \cite{forni_lippi_2001} (where however the characterization in terms of  boundedness of $\Sigma$ was not noticed).

\section{Purely Deterministic  sequences}\label{sect:Aggreg}
The notion of a (discrete-time) {\em purely deterministic}  random sequence, or process,  is well-known, see e.g. \cite{Cramer-61} who originated the terminology for not necessarily stationary processes.
 \begin{definition}
Let $q$ be a finite natural number. A sequence $\yb$ is  {\em  purely deterministic of rank $q$}
(in short $q$-PD)  if $H(\yb)$ has dimension $q$.
\end{definition}
Clearly   a $q$-PD sequence $\yb$ is a (in general non-stationary) purely deterministic process in the classical  sense of the term.  More specifically,  $\yb$ is a $q$-PD random sequence if and only if there are $q$ deterministic infinite column vectors $f_1,\,f_2,\,\ldots f_q$ or, for short, an      $\infty \times q$ matrix $F= \begin{bmatrix} f_1&f_2& \ldots& f_q\end{bmatrix}$, such that
\begin{equation}\label{qPDProc}
 \yb(k)\, = \sum_{i=1}^{q} \, f_i(k)\, \xb_i\,,\quad k\in \Zbb_+\,,\quad \text{or} \quad \yb=F\xb
\end{equation}
for some random variables $ \xb_1 ,\; \ldots ,\; \xb_q $. Without loss of generality, the columns $f_1,\,f_2,\,\ldots f_q$  can  be  assumed to be linearly independent, for otherwise one of them could be expressed as a linear  combination of the others and eliminated. In this case $\{\, \xb_1 ,\; \ldots ,\;   \xb_q\, \} $  can be taken to be an orthonormal basis in $H(\yb)$.

We want to relate this concept to the factor  component of $\yb$, as defined earlier.  The factors will later on be shown to originate the flocking component of $\yb(t)$ in the time varying case. As it stands the  $q$-PD condition is however insufficient to guarantee uniqueness.  Unfortunately it turns out that  there are nontrivial sequences representable in the form \eqref{qPDProc} which are idiosyncratic (or contain idiosyncratic sequences). See the the example below.
\begin{example}\label{Ex3}
Consider a sequence $\yb$ whose $k-$th element is
\begin{equation}
\yb(k) = \lambda^k \xb \quad,\,|\lambda| < 1,
\end{equation}
where $\xb$ is a zero--mean random variable of positive  variance $\sigma^2$.
Clearly, $\yb$ is  1-PD, its spanned subspace
$H(\yb)$ being  the one-dimensional space $H(\xb)$. The covariance
matrix of the first $n$ components of $\yb$  is
\begin{equation}
\Sigma_n = \mathbb{E}\yb_{n}\yb_{n}^{\top} = \sigma^2 \begin{bmatrix}\lambda^2 & \lambda^3 & \ldots & \lambda^{n+1}  \\ \lambda^3 & \lambda^4 & \ldots & \lambda^{n+2} \\ \vdots & \vdots &\ddots& \vdots \\
\lambda^{n+1}&\lambda^{n+2} &\ldots & \lambda^{2n} \end{bmatrix}
\end{equation}
Since $\rank(\Sigma_n) = 1$ for every $n$, we have
\begin{align}
\lambda_1(\Sigma) = \lim_{n\rightarrow \infty} \Tr(\Sigma_n) = \lim_{n\rightarrow \infty} \sigma^2\sum_{k=1}^n
\lambda^{2k} = \frac{\sigma^2\lambda^2}{1-\lambda^2} \,,
\end{align}
thus, in force of Theorem \ref{th:idiosyncratic}, $\yb$ is
idiosyncratic. Hence there are (non-stationary) $q-$PD sequences which are  idiosyncratic.
\end{example}

This is a possibility  which must clearly be  excluded if the decomposition \eqref{GFAD} has to be unique~\footnote{Note for example that Definition 2 in \cite[p.1294]{Chamberlain-R-83} is not enough to guarantee uniqueness.}. The question   is  which properties need to be satisfied by the functions $f_1,\,
f_2,\,\ldots f_q$ for $\yb$ not to be an idiosyncratic  sequence. One necessary condition is easily found: the $f_i$ {\em cannot be in $\ell^2$} since otherwise any sequence of functionals $\{a_n\}$ in $\ell^2$  converging    to zero   would lead to
\begin{equation}
\lim_{n\to \infty} \, a_n^{\top} f_i =0
\end{equation}
so that $\lim_{n\to \infty} \, a_n^{\top} \yb =0$ as well. This is clearly the problem with Example \ref{Ex3}.
\begin{proposition}\label{BddVar}
{\em If $\yb$ is   $q-$PD sequence with a uniformly bounded variance, then the   $f_i$'s  are uniformly bounded sequences; i.e. belong to the space $\ell^{\infty}$. If in addition  $\yb$  is not idiosyncratic the  $f_i$'s  belong to   $\ell^{\infty}$ but cannot belong  to $\ell^2$.}
\end{proposition}
\begin{proof}
The   statement  follows since $\| \yb(k)\|^2\leq M^2$, which is the same as $  \sum_{i=1}^{q} f_i(k)^2 \leq M^2$ implies that $|f_i(k)|\leq M$ for all $k$'s.
\end{proof}
We now discuss conditions in terms of the covariance matrix.
\begin{definition}
A  $q$-PD sequence $\yb$ is  {\em $q$-aggregate} if  $\rank\; \{ \Sigma\}  =q $  and \\ $\lim_{n \to \infty} \lambda_{k}(\Sigma_n) = +\infty$ for $k=1,\ldots,q$. In short, there are only $q$  nonzero eigenvalues of $\Sigma$ which are all  {\em   infinite}.
 \end{definition}
For $q=1$ this condition just means that the (only) column of $F$ has  $\ell^2$-norm equal to infinity.
 \begin{proposition}\label{Prop:uniqueness}
  {\em  A $q$-aggregate sequence  $  \yb$ can  be idiosyncratic only if it is the zero sequence.}
 \end{proposition}
 \begin{proof}
This follows trivially  from  Theorem \ref{th:idiosyncratic}. If $q>0$ the maximal eigenvalue of the covariance matrix of $ \yb$ is $+\infty$ by definition.
 \end{proof}
 Hence the condition guarantees some sort of uniqueness of the decomposition \eqref{GFAD}.
Of course the question is under what conditions the $q$ eigenvalues of $\hat{\Sigma}$ may tend to infinity. The notion of strong linear independence introduced  below provides an answer.
\begin{definition}
Let
\begin{equation}\label{eq:condition_tilde}
 \tilde{f}_i^n := f_i^n - \Pi[\, f_i^n \,|\,\mathcal{F}_i^n]
\end{equation}
where $\Pi$ is the orthogonal projection onto the  Euclidean space $ \mathcal{F}^n_i = \Span \{f_j^n ,\,j\neq i \, \} $ of dimension $q-1$.
The vectors $f_i, \,i=1,\ldots,q$ in $\Rbb^{\infty}$ are {\bf strongly linearly independent} if
\begin{equation}\label{eq:condition}
\lim_{n\rightarrow\infty}\|\tilde{f}_i^n\|_2 = +\infty\,\qq \,i=1,\ldots,q\,.
\end{equation}
\end{definition}
\medskip

In a sense, the tails of  two strongly linearly independent vectors in $\Rbb^{\infty}$ cannot get ``too close'' asymptotically.
 \begin{theorem}\label{Thm:Strong}
 {\em Let $\yb$ be a $q-$PD sequence, i.e. let
\begin{equation}
\yb(k)\,   = \sum_{i=1}^{q} \, f_i(k)\, \xb_i\,,\qquad
k\in \Zbb_+\,;
\end{equation}
then $\yb$  is   $q-$aggregate  if and only if, the vectors $f_i, \,i=1,\ldots,q$ are strongly linearly independent.}
\end{theorem}
The proof is given in   Appendix \ref{App:C}.

\begin{example}
Consider the $2-$PD sequence $\yb(k) := \sum_{i=1}^{2} \, f_i(k)\, \xb_i$, with
$$
f_1(k) = 1 \quad \mbox{for all } k \,,\quad\quad f_2(k) = 1 - \left(\frac{1}{2}\right)^k
$$
It is not difficult to check  that this sequence does not satisfy condition \eqref{eq:condition}. We shall show that this sequence is not  2-aggregate. The   Gramian matrix of the functions $f_1,f_2$
restricted to  $[1,\, n]$ is
$$
F^{n\top} F^n = \begin{bmatrix} \|f_1^n\|_2^2 & \langle f_1^n,\,f_2^n \rangle_2 \\
\langle f_1^n,\,f_2^n \rangle_2 & \|f_2^n\|_2^2 \end{bmatrix}
$$
and it can be  seen that as $n \rightarrow \infty$, the second eigenvalue converges to
$\frac{5}{3}$. Hence one eigenvalue of the covariance matrix of $\yb$ is finite and the sequence is not 2-aggregate. \hfill$\Box$
\end{example}

\section{GFA representations: Existence and uniqueness }\label{GFA}
Summing up  Theorem \ref{th:idiosyncratic}, Theorem \ref{Thm:Strong} and the uniqueness result in Proposition \ref{Prop:uniqueness}
we obtain conditions on the covariance $\Sigma = \E \yb \yb^{\top}$ to   describe processes admitting a GFA representation.
\begin{definition}
The covariance $\Sigma$ has a {\em GFA decomposition of rank $q$}   if it can be decomposed as the sum of a  matrix  $\tilde{\Sigma}$ which is a bounded operator  in $\ell^2$, and a  $\rank\; q$ perturbation $\hat{\Sigma}=F F^{\top}$, namely
\begin{equation}\label{CovGFA}
            \Sigma=F F^{\top} + \tilde{\Sigma}\,,\qq \text{with} \q F= \bmat f_1&\ldots &f_q\emat\,,\;\; f_i \in \Rbb^{\infty}
\end{equation}
where $F\in \Rbb ^{\infty \times q}$ has strongly linearly independent columns.
\end{definition}
 \begin{theorem} \label{CovCond}
 The infinite covariance matrix $\Sigma$ has a GFA decomposition of rank $q$   if  and only if   for $n\to \infty$, $\Sigma_n$ has $q$   unbounded eigenvalues   and $\lambda_{q+1}(\Sigma_n)$  stays bounded as $n\to \infty$. A GFA decomposition of $\Sigma$    is unique, modulo right multiplication of $F$ by a $q \times q$ orthogonal matrix.
\end{theorem}
This result is close to
Chamberlain and Rothschild \cite[Theorem 4]{Chamberlain-R-83} where it is  obtained via a quite different and rather lengthy series of arguments.

Note that there may well be  sequences (of positive symmetric) $\Sigma_n$ for which {\em all eigenvalues}  tend to infinity. In this case there is no GFA decomposition.    When it applies, the criterion can be seen as a limit of the well-known rule of separating ``large''  from ``small'' eigenvalues in Principal Components Analysis (PCA).   Let $f^{n}_{i} \in \Rbb^{n}\,;\, i=1,\ldots,q$ be the eigenvectors     corresponding to the   $q$  (ordered) eigenvalues of $\Sigma_n$ which increase without bound when $n\to \infty$. We normalize these eigenvectors in such a way that   $F_n:= \bmat f^{n}_{1} & \ldots & f^{n}_{q}\emat$ yields $\hat{\Sigma}_n = F_nF_n^{\top}$. Then
\begin{equation}
\lim_{n\to \infty} F_nF_n^{\top}= F F^{\top}\,.
\end{equation}
See \cite[Theorem 4 (ii), p. 1299] {Chamberlain-R-83} for a proof and a discussion. Then convergence of $\{F_n\}$ can be interpreted as column space convergence in the gap metric, see \cite[p. 260] {Stewart-S-90}.  Although the usual  orthogonality of the  $f^{n}_{i}$  in PCA does not make sense in infinite dimensions as the limit eigenvectors do not belong to $\ell^2$, one may however interpret the strong linear independence condition   as a limit of the orthogonality holding for finite $n$. Hence we can (asymptotically)  get $q$ and $F$ by a limit  PCA procedure      on the sequence  $ \Sigma_n$.

 Trivially, if a random sequence $\yb$ admits a GFA representation then its covariance matrix has a GFA decomposition. On the other hand, assume we are given a GFA decomposition $\hat{\Sigma}+\tilde{\Sigma}$ of an infinite  covariance $\Sigma$.
 How do we find the hidden variables in the representation $\yb= F\xb +\tilde \yb$?\\
 This question has also to do with uniqueness of the representation as there may be several non-equivalent choices of $\xb$ and $\tilde \yb$ compatible with a GFA decomposition of $\Sigma$. We shall show that there is an essentially unique choice, under  the constraint   that both $\xb$ and $\tilde{\yb}$ belong to $H(\yb)$. Models of this kind are called {\em internal} in stochastic realization.
 The following definition  from \cite{forni_lippi_2001}  is  meant to generalize  the phenomenon described in Example \ref{Stringof1}.

\begin{definition}
  Let $\zb \in H(\yb)$. The random variable $\zb$ is an {\em aggregate (of $\yb$)} if there exists
an AS $\{a_n\}$  such that $\lim_{n\rightarrow\infty}  a_n^{\top}\yb =\zb$.
The set of all aggregate random variables in $H(\yb)$   is a closed subspace denoted by $\mathcal{F}(\yb)$ called  the {\em aggregation   subspace} of $H(\yb)$.
\end{definition}
Clearly,  if $\yb$ is an idiosyncratic sequence then
$\mathcal{F}(\yb) = \{0\}$. One can then define an orthogonal decomposition of the type
\begin{equation}
\yb = \mathbb{E}[\yb \mid \mathcal{F}(\yb)] + \ub \,,
\end{equation}
where $\mathbb{E}[\,\cdot \,\mid \mathcal{F}(\yb)]$ is the orthogonal projection operator onto the subspace $\mathcal{F}(\yb)$, so that all components $ \ub(k)$ are uncorrelated with
$\mathcal{F}(\yb)$. The idea behind this decomposition is that, in
case $ \mathcal{F}(\yb)$ is finite dimensional, say generated by a
$q$-dimensional  random vector $\xb$,  one may naturally identify $\ub$ as the idiosyncratic component and capture  a
unique  decomposition of $\yb$ of the type \eqref{GFAD}. This intention is probably behind the analogous decomposition in \cite{forni_lippi_2001} but this idea cannot be pursued further unless some further  technical requirements    are imposed, which are so far unknown.  There may be pathological situations in which $ \mathcal{F}(\yb)$ is finite dimensional, or in which  $\mathcal{F}(\yb) = \{0\}$, but the process $\ub$  is not idiosyncratic.
Theorem \ref{MainThm} below asserts that  in the special case of    stationary sequences,  the construction works if and only if its spectral density is in $L^{\infty}$.

 \begin{proposition}\label{Prop:Constr}
  Assume that its covariance matrix  $\Sigma$ has     a GFA decomposition of rank $q$. Then $\yb$ has a   GFA representation with $q$ factors where both $\xb$ and $\tilde \yb$ have components in $H(\yb)$.
 \end{proposition}
 \begin{proof}
   By a standard Q-R factorization we can orthogonalize the columns of $F_n$,
\begin{align}
 \bmat f_1^{n} &f_2^{n}& \ldots&f_q^{n}\emat  =& \\ 
 &\!\!\!\!\!\!\!\!\!\!\!\!\!\!\!\!\!\!\!\!\!\!\!\!\!\!\!\!\!\!\!\!\!  \bmat  g_1^n & g_2^n & \ldots & g_q^n\emat \bmat 1& r_{1,2}& r_{1,3}&  \ldots & r_{1,q}\\
                                    0 & 1 & r_{2,3}& \ldots & r_{2,q}\\
                                      0 & 0 &1 & \ddots & r_{3,q} \\
                                    \dots & \dots & \dots &\ddots &\dots \\
                                     0 & 0 & 0& \ldots & 1\emat  \nonumber
\end{align}
which we shall write compactly as
\begin{equation}\label{Q-Rfact}
F_n = Q_n R_n
\end{equation}
where $Q_n := \begin{bmatrix} g_1^n &  g_2^n & \ldots &  g_q^n  \end{bmatrix}$    has orthogonal columns. It is well-known that each $g_{i}^n$  can be obtained by a sequential Gram-Schmidt orthogonalization procedure as the difference of  $f_i^{n} $ with its projection onto the subspace $ \Span \{f_j^n ,\,j< i \, \}   \subset   \mathcal{F}^n_i $. Hence $\|g_{i}^n\| \geq \|\tilde{f}_i^{n}\|$ and hence, by assumption, tends to $\infty$ when $n\to \infty$.\\
Next, define
 \begin{equation}
 a_{i,n}^{\top}:= \Frac{1}{\| g_i^n\|_2^2}\, \bmat  g_i^n(1) &  g_i^n (2) & \ldots& g_i^n(n) & 0& \ldots\emat
 \end{equation}
 where the $g_i^n$'s are as defined above. Since $\|g_i^n\|_2 \to \infty$ with $n$, we have    $\| a_{i,n}\|_2= 1/\|g_i^n\|_2 \, \rightarrow 0$ as $n \to \infty$. Hence $a_{i,n}$ is an AS.\\
 Note that we can express each $f_i^{n}$ as $ f_i^{n} = g_i^n +\sum_{j=1}^{i-1} r_{j,i}  g_j^n$  so that
 \begin{equation}
  a_{i,n}^{\top} f_i=  \Frac{1}{\| g_i^n\|_2^2}\,\| g_i^n\|_2^2 \,=\,1
 \end{equation}
  for all $n$ large enough and by a similar calculation one can easily check that $a_{i,n}^{\top} f_j=  0 $, for all $j<i$. With these $a_{i,n}$ construct a sequence of $q \times \infty$ matrices
 \begin{equation}
 A_n := \bmat a_{1,n}^{\top} \\ \dots \\  a_{q,n}^{\top} \emat
 \end{equation}
 which  provides     an asymptotic left-inverse of $F$, in the sense that $\lim _{n \to \infty}\, A_n F= R$, where $R $ is the limit of a sequence of   $q\times q$ matrices  all of which are upper triangular  with ones on the main diagonal. Next,  define  the random vector  $\zb_n := A_n \yb$ which  converges  as $n \to \infty$ to a $q$-dimensional  $\zb$ whose   components must belong to $ \mathcal{F}(\yb)$.  These $q$ components form in fact a basis for $ \mathcal{F}(\yb)$ as the covariance $\E \zb_n \zb_n^{\top}$ converges to $RR^{\top}$ which is non singular. From this, one can easily get  an orthonormal basis $\xb$,  in $H(\hat \yb)$. Hence, since $F$ is known, we can form $\hat \yb = F \xb$ and letting $\tilde \yb:= \yb - \hat \yb$ does yield a GFA representation of $\yb$ inducing the given GFA decomposition of $\Sigma$. Uniqueness is then guaranteed in force of Proposition \ref{Prop:uniqueness}.
 \end{proof}
 \subsection{Interpretation: Short and long distance interaction}\label{SubsectA}
 Imagine a scenario of    an ensemble of infinitely many agents distributed in space interacting   randomly, producing as output measurements  the random variables $ \yb(k)= \hat{\yb}(k) + \tilde{\yb}(k)\;;\; k=1,2,\ldots$.\\
  The covariances $  \tilde\sigma(k,j)= \E\tilde\yb(k) \tilde\yb(j)\, $ measure  the mutual correlation of  the idiosyncratic  fluctuations of neighboring  agents   $\tilde\yb(k),\,\tilde\yb(j)$   located in positions $k$ and $j$. Since $ \tilde{\Sigma}$ is a bounded operator in $\ell^2$, it is a known fact \cite[Section 26]{Akhiezer-G-61} that  $ \tilde\sigma(k,j)\rightarrow 0$ as $|k-j|\rightarrow \infty$ so, in a  sense the idyosincratic component $\tilde\yb$ of a GFA representation models only {\em short range } interaction among the agents, as $\tilde\sigma(k,j)$ is decaying to zero when the   distance $|k-j|$ tends to infinity.

   Whenever an  ensemble can be described by an idiosyncratic sequence, then agents which are far away from each other essentially do not resent of mutual influence.   The statement    holds in general, for every GFA model, although the  decay of the elements $  \tilde\sigma(k,j)$ may be faster depending on the     particular covariance structure. Just the opposite will be  true for the sequence $\hat {\yb}$.

  On the other hand,   $\E \hat{\yb}(k)\hat{\yb}(j)= \sum_{i} f_i(k)f_i(j)$ and  the elements of the column vectors  $f_i  $ cannot be in $\ell^2$. In particular, as stated in Proposition \ref{BddVar}, $f_i \in \ell^{\infty}$ when the  variances of the random variables $\yb(k)$ are uniformly bounded.\\
In any case,  since the components $f_i(k)$ do not decay with distance, the products $f_i(k)f_i(j)$ generically cannot vanish when $|k-j|\rightarrow \infty$. Therefore     the factor loadings describe  ``long range'' correlation and the aggregate component $\hat{\yb}$ of $\yb$ can be interpreted as variables modeling    {\em long range interaction} among the agents. In this sense $\hat{\yb}$ models an average  {\em collective behavior}  of the ensemble. This is in fact the core of the flocking structure that  will emerge  as soon as the $\xb_i$ are allowed to depend on time.

\subsection{The case of stationary sequences}\label{Sec:Wold}
The characterizations of GFA models discussed   so far are for general second order sequences, that is for processes $\yb$ which may well be non-stationarity with respect to the cross  sectional (space) index $k$.  Much sharper results hold in  the special case in which the  sequence  $\yb$
  is   (weakly) stationary; i.e. $\mathbb{E}\yb(k) \yb(j) =\sigma(k-j)$ for $k,j \geq 0$. A complete analysis of this case cannot be presented here and can be found in \cite{Bottegal-P-Arkiv13}. Here we shall just report the main result.

   Let $H_{k}(\yb)$ be the closed linear span of all    random variables $\{\yb(s)\,;\, s \geq k\}$. Introducing the {\em remote future subspace} of $\yb$:
\begin{equation}\label{InfFut}
H_{\infty}(\yb) = \bigcap_{k \geq 0} H_{k}(\yb)\,,
\end{equation}
   the sequence of orthogonal wandering subspaces $E_k := H_{k}(\yb) \ominus H_{k+1}(\yb)$ and their orthogonal direct sum
$
\check{H} (\yb)= \bigoplus_{k\geq 0}\, E_k\,,
$
 it is well known, see e.g. \cite{Doob-53,Halmos-61,Rozanov-67}, that  one has the orthogonal decomposition
\begin{equation}\label{InfFut2}
 \yb  =\hat {\yb} + \check{\yb}\,,\qquad \hat {\yb}(k) \in H_{\infty}(\yb)\,\q \check{\yb}(k) \in \check{H} (\yb)
\end{equation}
for all $k \in \Zbb_{+}$, the component $ \hat {\yb}$ being the purely deterministic (PD), while $\check{\yb}$  the purely non-deterministic (PND) components. The two sequences are orthogonal and uniquely determined.
Furthermore, it is well known that $\check{\yb}$ has an absolutely continuous spectrum with a spectral density function, say $S_y(\omega)$ satisfying the log-integrability condition $\int \log S_y(\omega)\, d\omega> -\infty$, while the spectral distribution of $\hat {\yb}$  is singular  with respect to Lebesgue measure (for example consisting only of jumps) possibly together with   a   spectral density such that $\int \log S_y(\omega)\, d\omega= -\infty$, compare e.g. \cite{Rozanov-67}.

\begin{theorem}\label{MainThm}
 {\em   Assume that $\yb$ is a stationary sequence with   $\dim
H_\infty(\yb)~< ~\infty$ and an a.e. bounded spectral density. Then $H_\infty(\yb) \equiv \mathcal{F}(\yb)$.

  A stationary sequence    admits a unique internal GFA representation \eqref{GFAD} with $q$ factors if and only if it has a bounded spectral density and  the remote future space is of dimension $q$. The aggregate component  $\hat \yb$ is the purely deterministic component of $\yb$ while the idiosyncratic  $\tilde \yb$ is the purely non-deterministic component.}
\end{theorem}

 Note that  there are stationary processes with a finite dimensional  remote future space, whose PND component has  an
unbounded spectral density. It follows from Szeg\"o's theorem that $\tilde \Sigma$ is an unbounded operator and these processes are neither aggregate nor idiosyncratic.

In the   papers \cite{Chamberlain-R-83,forni_lippi_2001}, stationarity with respect to the cross-sectional  index is not assumed. However without stationarity, there may be random sequences which fail to satisfy the eigenvalue conditions of Theorem \ref{CovCond} and do not admit a generalized factor analysis representation.   A precise characterization of  which class of non-stationary sequences   admits  a GFA representation seems still to be an open problem.

\section {Dynamic GFA models}\label{Sec:RandFields}
We come back to dynamic modeling and to the question raised in section \ref{ProblemSect} namely  when does a second order random field have a flocking component. We shall initially restrict to the case of processes which are stationary with respect to the time variable which  is a natural assumption to make in view of statistical inference.

A time-dependent family $\yb:= \{\yb(t)\,;\, t \in \Zbb \}$ of infinite-dimensional  zero-mean random vectors, $\yb(t)$,  whose covariance matrix, $\Sigma(\tau) := \E \yb(t+\tau) \yb(t)^{\top}$  is  (finite and) independent of $t$, will be called   a {\em time-stationary (second order) random field}.  The following  definition extends and makes precise the finite-dimensional concepts introduced at the end of  Section \ref{Introd}.

\begin{definition}
We shall say that a time-stationary random field has a {\em dynamic GFA representation} of rank $q$ if it can be written as
\begin{equation}\label{DGFA}
\yb(t)= F \xb(t) + \tilde{\yb}(t)
\end{equation}
where $F \in \Rbb^{\infty \times q}$ has strongly linearly independent columns and
$ \tilde{\yb}(t)$ is an {\em idiosyncratic random field}; i.e   the covariance matrix $ \tilde {\Sigma}(\tau):= \E \tilde{\yb}(t+\tau) \tilde{\yb}(t)^{\top}$ defines, for all $\tau\in \Zbb$, a bounded linear operator in $\ell^2$. The $q$ dimensional {\em factor process}   $\xb(t)$ and $ \tilde{\yb}(t)$  are jointly stationary and uncorrelated, that is
$$
\E\xb_i(t)\,\tilde{\yb}_j(t) = 0\,,\qquad i=1,\ldots,q;\; j= 1,2, \ldots, \q t \in \Zbb\,.
$$
Without loss of generality,    $\xb(t)$  can be chose with orthonormal components; i.e.   $\E \xb(t)\xb(t)^{\top}=~I_q$.
 \end{definition}
 \begin{proposition}\label{DynStationary}
  The stationary random field $\yb:= \{\yb(t)\,;\, t \in \Zbb \}$ has a dynamic GFA representation \eqref{DGFA} if and only if  $\yb(0)$   has a static GFA represenattion with  the same factor loading matrix $F$,  $\xb\equiv \xb(0)$ and $ \tilde{\yb} \equiv \tilde{\yb}(0)$.
\end{proposition}
\begin{proof}
 The proof of the direct implication is trivial. The converse  is proven in Appendix~\ref{App:E}.
\end{proof}
Incidentally, the proposition guarantees uniqueness of the dynamic representation \eqref{DGFA}.  The following criterion for the existence of a flocking structure in a time-stationary random field   follows directly from Theorem  \ref{CovCond} and the proposition above. 
\begin{corollary}
{\em For a time-stationary random field, a flocking structure exists with  $q$ factors    if  and only if $q$ eigenvalues of the steady state covariance matrix $\Sigma_n$ of the  $n$-dimensional random subvector  $\yb^n(t)$ of $\yb(t)$, tend to infinity with $n$  while the others remain bounded.}
\end{corollary}

We shall study a  class of   random fields   described by linear evolution equations of the general form
\begin{equation}\label{LinEvol}
\yb(t+1)= A\yb(t) +\wb(t)\,\; \qquad t\in \Zbb
\end{equation}
where $\wb$ is a string of uncorrelated stationary white noise processes and $A$ is an infinite matrix (a linear operator) acting on infinite sequences. We assume that the evolution is  stationary in time so that the steady state covariance matrix of $\yb(t)$ is a constant positive definite matrix $\Sigma$, which should  satisfy an infinite dimensional Lyapunov equation
\begin{equation}\label{LYAP}
\Sigma =  A \Sigma  A ^{\top} +Q
\end{equation}
where   $Q$ is the variance matrix of the white noise which we assume  an infinite  diagonal matrix with uniformly bounded positive entries (it is actually no loss of generality assuming that $Q$ is the identity matrix).  In this case, a GFA model of $\yb$  will also be stationary and the  structure of the model can be inferred by analyzing the covariance matrix $\Sigma$. To this end   consider the   $n$-dimensional random sub-processes $\yb^n(t)$ of $\yb(t)$, obeying the  equation 
\begin{equation} \label{eq:basic}
\yb^n(t+1) = A_n \yb^n(t) + \wb^n(t) \,,\quad n=1,2, \ldots
\end{equation}
where  $A_n$ is the upper left $n \times n$ submatrix of $A$ and the input process
$\wb^n(t)$ is the $n$-dimensional white noise with  variance  $ \mathbb{E} \wb^n(t) \wb^n(s)^\top = Q_n\delta_{t,s} $, $Q_n$ being the the upper left $n \times n$ submatrix of $Q$, and  study the behavior  as $n\to \infty$ of   the  covariance matrix of $\yb^n(t)$, solution to the family of Lyapunov equations
\begin{equation} \label{eq:fam_lyap}
\Sigma_n = A_n\Sigma_nA_n^\top + Q_n\,\qquad n=1,2,\ldots\;.
\end{equation}
The existence of a flocking component  can be  addressed by analyzing the asymptotics   of $ \Sigma_n$ when   $n\to \infty$. 
Some    families of matrices $\{A_n\}_{n \in \mathbb{N}}$ are considered below.

   \subsection{Autonomous agents}
In this scenario, the behavior of each agent is independent of  the others, being just an autoregressive motion of the type
\begin{equation}
\yb_k(t+1) = a_k \yb_k(t) + \wb_k(t)  \quad, \quad \sup_{k\in \mathbb{N}} |a_k| < 1     \,.
\end{equation}
In this case, $A_n = \diag \{a_1,\,\ldots,\,a_n\}$ and the family of Lyapunov equations \eqref{eq:fam_lyap} admits diagonal (nested) solutions with uniformly bounded elements. Hence, in this case the resulting  sequence is   idiosyncratic noise with uncorrelated components and there is no flocking structure.

\subsection{ Flocking by following a leader}
As discussed in    \cite{Shen-07}, flocking may be observed  in   hierarchical leadership models where the evolution of the   first $n$ agents   influences that of the agents of index $k>n$ but not conversely so the matrix of the operator $A$ has a nested lower triangular structure   of the type
\begin{equation} \label{eq:nested}
A_{n+1} = \begin{bmatrix} A_n & 0 \\ b_n^\top & a_{n+1} \end{bmatrix} \,,
\end{equation}
where  $|a_{n+1}|<1$ to keep the asymptotic stability of $A_{n}$  preserved.  

A very simple instance is the following linear model where each  agent, evolving  with the same scalar random dynamics, wants to follow a leader by applying a  proportional control law based on the measurement of its position with respect to the leader's $\yb_0(t)$:
 \begin{eqnarray*}
\yb_1(t+1) &=& a  \yb_1(t) + \wb_1(t)\,,\qquad |a|<1\\
  \yb_k(t+1) &=& (1-a)  \yb_1(t) +a  \yb_k(t) + \wb_k(t)\,,\, k=2,3,\ldots
\end{eqnarray*}
 The question  is if  following  a leader should, under appropriate circumstances,  produce a random flock. The steady-state covariance matrices of  $\yb^{n}(t)$  solves the Lyapunov equation \eqref{eq:fam_lyap} for the   model  
$$
\bmat \yb_0(t+1) \\ \yb_1(t+1) \\ \dots \\ \yb_n(t+1)\emat = \bmat a& 0 & \ldots&0\\
1-a &a &  & \vdots \\ \vdots & 0 & \ddots& 0\\ 1-a & \ldots& & a \emat \bmat \yb_0(t)\\ \yb_1(t) \\ \dots \\ \yb_n(t)\emat +\bmat  \wb_0(t)\\ \wb_1(t) \\ \dots \\ \wb_n(t)\emat
$$
and it is possible to show  that  indeed  a flocking structure is present. 
  \begin{proposition}\label{thm:leader}
{\em Assume for simplicity   that $Q_n=I_n$. The solution of the Lyapunov equation \eqref{eq:fam_lyap}  tends for $n\to \infty $ to
a covariance matrix of the form $ \Sigma= f f^{\top} + \tilde{\Sigma} $ where $f \in \Rbb^{\infty}$ has components
$$
f_k= \begin{cases} \begin{array} {lcr} a/(1-a^4)^{\frac{1}{2}} & ,& k=1\\ (1+a^2)^{\frac{1}{2}}/[ (1+a)(1-a^2)^{\frac{1}{2}}] & ,& k>1, \end{array} \end{cases}
$$
and $ \tilde{\Sigma}$ is a bounded operator in $\ell^2$. Hence
$$
\yb(t)= f \xb(t) + \tilde{\yb}(t)\,,\quad    \xb(t) = (1-a^4)^{\frac{1}{2}} \yb_1(t-1) \,,  
$$
$$
\Var \tilde{\yb}(t)= \tilde{\Sigma}\,.
$$}
\end{proposition}
The calculations and  the structure of $\tilde{\Sigma}$ are  in Appendix \ref{App:F}. Note that the infinite matrix $A$ does not define a bounded operator on the whole space  $\ell^2$ since the first column is not square summable (it just belongs to $\ell^{\infty}$) the domain being the linear subspace of all sequences in $\ell^2$  having  zero initial symbol.

\subsection{Infinite dimensional distributed average consensus}
Assume that the $k$-th agent adjusts its output  in discrete time by a  symmetric linear relation
 \begin{equation}  \label{AdjustCons}
 \yb_{k}(t+1) = a_k \yb_{k}(t)+\sum_{j\in N_k}a_{k,j}(\yb_{j} (t) - \yb_{k}(t)) + \wb_k(t)\,,
 \end{equation}
  where  $k=1,\,2,\,\ldots$ and the sum is over the set of neighbors $N_k$ of each state $k$, which we assume to be a finite set.   The overall motion can be described as
 \begin{equation} \label{Amatrix}
 \yb(t+1)= A \yb(t) +\wb(t)
 \end{equation}
starting at some initial state $\yb(0)$. Here  $A$ is a matrix with positive elements such that
$$
A=A^{\top}\, \qquad  A {\mathbf 1}\!\!1={\mathbf 1}\!\!1
$$
an infinite doubly stochastic matrix. The state of \eqref{Amatrix} is not stationary since it  has a random walk component. We want to see if for some averaging sequence $\{a_n\}$ the limit
$$
\lim_{n\rightarrow \infty} a_n^{\top}\xb(t)
$$
is non-zero. This would imply the existence of a flocking component. Problems of this kind have been studied in the finite-dimensional setting in \cite{Xiao-B-K-07}. Here we study a slightly different model, obtained by modifying \eqref{AdjustCons}  so as to deal with an infinite number of agents. Let us assume that
\begin{enumerate}
\item for each $n \geq n_0$, where $n_0$ is a fixed initial integer,   the symmetric doubly stochastic matrix  $A_n$,  achieves consensus on the first $n$ agents;
\item  define  a sequence of matrices $\bar A_n := (1- \frac{1}{n})A_n$,   and assume that consensus is reached as $n \rightarrow \infty$.
\end{enumerate}
Denoting by $\bar A$ the limit of the sequence $\{\bar A_n,\, n \in \mathbb N \}$, the following result holds.
\begin{proposition} \label{thm:consensus}
{\em The model  \begin{equation} \label{Amatrix2}
\yb(t+1)= \bar A \yb(t) +\wb(t) \quad ,\, Q = I
\end{equation}
admits a flocking structure. The relative GFA decomposition has one ($q = 1$) latent factor.}
\end{proposition}
The proof is in Appendix \ref{App:G}.

\subsection{Generalizations} 
In both the above examples the matrix $A$ can be decomposed as the sum of a bounded operator in $\ell^2$ plus an unbounded  rank one perturbation in $\ell^{\infty}$. It then happens that the unbounded solution of the Lyapunov equation has exactly the same columnspace and the same rank   as the unbounded perturbation of $A$. Although we don't have  a general  proof, this seems likely to be a general fact.  Let us assume, by way of example, that $A$ is  symmetric and is a direct sum of a finite, rank $q$ perturbation plus a bounded operator in $\ell^2$, defined on $\mathscr{F} \oplus \ell^2$ where $\mathscr{F}$ has dimension $q$. Because of symmetry  these will be  orthogonal  complementary invariant subspaces. For each finite $n$ we   therefore have a block decomposition
$$
A\bmat F & G\emat= \bmat F & G\emat \diag\,\{ \hat{A},\, \tilde{A}\}
$$
where for $n\to \infty$  the $q$ columns of $F$   belong to $\ell^{\infty}$ but not to $\ell^2$  while $G$ is a unitary operator in  $\ell^2$. Writing formally $T:=  \bmat F & G\emat $ we have $A=T \diag\,\{ \hat{A},\, \tilde{A}\} T^{-1}$ and also
$$
A^{k}= T \diag\,\{ \hat{A}^{k},\, \tilde{A}^{k}\} T^{-1}
$$
Then, letting  $T^{-1}Q T^{-\top}:= \diag \,\{ \hat{Q} ,\, \tilde{Q}\}$,  the solution of \eqref{LYAP} can be written
\begin{align}
\Sigma & =  \bmat F & G\emat \diag\,\left\{\sum_{k=0}^{+\infty} \hat{A}^{k} \hat{Q} [\hat{A}^{\top}]^k ,\;  \sum_{k=0}^{+\infty} \tilde{A}^{k} \tilde{Q} [\tilde{A}^{\top}]^k \right\} \bmat F^{\top} \\ G^{*}\emat\nonumber \\
&= F \hat{P} F^{\top} + G\tilde{P} G^{*}
\end{align}
Hence when $A$ has $q$ eigenvectors in $\ell^{\infty}$ (but not in  $\ell^2$) the steady-state covariance has a GFA decomposition. \\
Changing basis in \eqref{LinEvol} by letting $\yb(t) = \bmat F & G\emat \bmat \hat{\xb}(t)\\ \tilde{\xb}(t) \emat$ so that
$$
\bmat \hat{\xb}(t+1)\\ \tilde{\xb}(t+1) \emat= \diag\,\{ \hat{A},\, \tilde{A}\}\bmat \hat{\xb}(t)\\ \tilde{\xb}(t) \emat + \bmat \hat{\wb}(t) \\ \tilde{\wb}(t)\emat
$$
we end up with a GFA decomposition $\yb(t)= \hat{\yb}(t) + \tilde{\yb}(t)$ where the two components 
$$
\hat{\yb}(t)= F\hat{\xb}(t),;\qquad \tilde{\yb}(t)= G \tilde{\xb}(t) 
$$
are the flocking and the idiosyncratic parts of $\yb(t)$. Note that   the noise components  $\hat{\wb}(t) $ and $ \tilde{\wb}(t)$ are mutually uncorrelated and hence so are $ \hat{\yb}(t)$ and $\tilde{\yb}(t)$.

\subsection{Separable space-time processes}

Random fields which are often encountered in geostatistics, hydrology, marine wave models, meteorology and environmental applications, see e.g \cite{Chunsheng-07} and the references therein,
belong to the class of so-called {\em separable space-time processes}
\begin{equation}\label{multRF}
\yb(k,t)= \sum _{i=1}^{m}\vb_{i}(k) \ub_{i}(t)
\end{equation}
represented as  the product of a  space, $\vb(k):= [\vb_{1}(k)\; \vb_{2}(k)\;\ldots \vb_{m}(k)] $,  and time component, $\ub(t):= [  \ub_{1}(t)\; \ub_{2}(t)\; \ldots  \ub_{m}(t)]^{\top}$, both zero mean and with finite variance.   In general one should take $m=\infty$ \cite{Venturi-2010} but finite dimensional approximations are often enough. To discuss these models one should  generalize the static theory in the preceding sections to $m$-vector-valued processes. Although this is quite straightforward, involving no new concepts but just more notations, for the sake of clarity we shall restrain to the scalar case $m=1$. \\
The model \eqref{multRF} needs to be specified probabilistically, as the dynamics of the ``time'' process $\{\ub(t)\}$ may well be space dependent and dually, the distribution of  $\vb(k)$ may be a priori time-dependent.
The following assumption specifies  in probabilistic terms the multiplicative structure \eqref{multRF} of the random field $\yb(k,t)$.
\medskip

 \noindent{\em Assumption:}   The space and time evolutions of $\yb(k,t)$  are {\em multiplicatively uncorrelated}  in the sense that
\begin{equation}\label{condExp}
\E \{  \vb(k_1)  \vb(k_2)\mid  \ub(t_1)\ub(t_2)\}= \E_{\vb}\{  \vb(k_1)  \vb(k_2)\}
\end{equation}
where the first conditional  expectation is made with respect to the conditional probability distribution of $\vb$ given the random variables $\ub(t_1),\,\ub(t_2)$, while the second expectation is with respect  to the marginal distribution of $\vb$.

From  the multiplicative uncorrelation \eqref{condExp} one gets
\begin{align}\label{ProdCov}
&\E\{  \vb(k_1)  \vb(k_2) \ub(t_1)\ub(t_2)\}= \E\{  \vb(k_1)  \vb(k_2)\} \, \E\{ \ub(t_1)\ub(t_2)\} \nonumber \\
& = \sigma_{\vb}(k_1,k_2) \,  \sigma_{\ub}(t_1,t_2)
\end{align}
where   $\sigma_{\vb}$ and $\sigma_{\ub}$ are the covariance functions of the two processes. Hence the covariance function of the random field inherits the separable  structure of the process. If $\vb$ and $\ub$ are jointly Gaussian, the multiplicative uncorrelation property follows if the two components are uncorrelated; namely their joint covariance is separable. This is a structure which is often assumed in the literature, see  \cite{Li-G-S-08} and references therein. Assume now that the space process has a nontrivial GFA representation with $q$ factors
\begin{equation}\label{SpaceGFA}
\vb(k)= \sum_{i=1}^{q}\, f_{i}(k) \zb_i + \tilde{\vb}(k)
\end{equation}
  where $\hat \vb(k):=\sum_{i} f_{i}(k) \zb_i  $ is the aggregate   and $\tilde \vb(k)$ the idiosyncratic  component of $\vb(k)$. Then setting $\xb_{i}(t)=  \zb_i \ub(t)$ and  $ \tilde{\yb}(k,t):= \tilde{\vb}(k)\ub(t)$ one can represent the random field  \eqref{multRF} by a dynamic GFA model,
\begin{equation}\label{TimeGFA}
\yb(k,t)= \sum_{i=1}^{q}\, f_{i}(k) \xb_i(t) + \tilde{\yb}(k,t):= \hat{\yb}(k,t) + \tilde{\yb}(k,t)
\end{equation}
\begin{proposition}
{\em If the processes $\vb$ and $\ub$ are multiplicatively uncorrelated then the two terms $\hat{\yb}(k,t)$ and $\tilde{\yb}(h,s)$  in the GFA model  \eqref{TimeGFA} are   uncorrelated for all $k, h$ and $t,s$. Hence a separable random field satisfying the  multiplicative uncorrelation property has a flocking component if and only if its space process $\vb$  has a nontrivial aggregate component.}
\end{proposition}
\begin{proof}
We have
\begin{equation}
\E  \{ \hat{\yb}(k,t) \tilde{\yb}(h,s)\}=  \sum_{i=1}^{q}\, f_{i}(k) \, \E \{ \zb_i \ub(t)\tilde{\vb}(h)\ub(s)\}
\end{equation}
where the expectation in the last term can be written as
\begin{align}
& \E \{ \zb_i \tilde{\vb}(h)\ub(t)\ub(s)\}= \E \{\E_{\vb} [  \zb_i \tilde{\vb}(h)\mid \ub(t)\ub(s)]\,\ub(t)\ub(s)\} \nonumber \\ 
& = \E \{\E_{\vb} [  \zb_i \tilde{\vb}(h)\,]  \ub(t)\ub(s)\}=0
\end{align}
since the $\zb_i$'s are random variables in $H(\hat{\vb})$ and $ \tilde{\vb}(h)$ is orthogonal to this space. The last statement then follows directly.
\end{proof}
 Here is probably the simplest nontrivial example of decomposition \eqref{TimeGFA}.  
 \begin{example}[Exchangeable space processes]
Consider  a {\em (weakly) exchangeable} space process  $\vb$; i.e. a process whose second order statistics are invariant with respect to all  index permutations  of locations $(k,j)$. Clearly the covariances  $\sigma_{\vb}(k,j)=\E \vb(k)\vb(j)$ must be independent of $k,\,j$ for $k\neq j$ and $\sigma_{\vb}(k,k)= \sigma^2>0$ must be independent of $k$, see \cite{Aldous-85}. Letting $\rho:= \sigma_{\vb}(k,j),\,k\neq j$, one has
\begin{equation}
\Sigma_{\vb}= \bmat \sigma^2& \rho& \rho& \rho &\ldots\\
            \rho & \sigma^2 & \rho& \rho &\ldots\\
            \ldots &          & \ddots&        &\ldots  \emat
\end{equation}
where   $\sigma^2 > |\rho|$ for positive definitness. Letting $f$ denote an infinite column vector with components all equal    to $\rho$, one can decompose $\Sigma_{\vb}$ as
\begin{equation}
\Sigma_{\vb} = f f^{\top} + (\sigma^2 - \rho) I
\end{equation}
where   $I$ denotes an infinite identity matrix. This is a Factor Analysis decomposition of rank $q=1$   of $\Sigma_{\vb}$ with $\tilde {\Sigma}_{\vb}$ a diagonal matrix. Hence a weakly     exchangeable space process is a 1-factor process with an idiosyncratic component which is actually white. In the GFA representation \eqref{SpaceGFA} there is just one factor $\zb$ and the factor loading vector $f$ does not depend on the space coordinate.  \hfill $\Box$ \\
Consider a random field with the multiplicative structure \eqref{multRF}, then the flocking component
$$
\hat{\yb}(k,t)= f \xb(t)\,,\qq \xb(t) =\zb \ub(t)
$$
describes a constant, space independent, configuration moving randomly in time.

\end{example}

\section{Conclusions}
 We have proposed a new modeling paradigm for large dimensional aggregates of random systems based on  the theory of Generalized Factor Analysis. We have  discussed in some depth  static GFA representations and characterized in a rigorous way their properties, especially the nature of  the   {\em idiosyncratic} and  {\em aggregate} components and provided new conditions guaranteeing uniqueness of the representation.  We have shown that the model   splits the output $\yb$ of the system into   two components describing the short- and long- range interaction among  the agents of the ensemble.  For wide-sense  stationary ensembles the nature and existence of these components can be   clarified in the light of the Wold decomposition. For time-dependent evolutions  the aggregate component  provides the core structure of the (random) flocking component.  A detailed analysis of interesting classes of random fields, such as the linear evolution equation in \eqref{LinEvol}, by using the decomposition of the steady state covariance has just been touched upon. Visibly, there is here ample room for further research on specific structures. Also   the   statistical identification had regrettably to be left out and will be considered in forthcoming publications.

 \appendix

\subsection{Proof of Lemma~\ref{th:Boundedness}}
Let $\lambda_1(\Sigma_n)$ be the maximal eigenvalue of $\Sigma_n$.
Since
\begin{equation}
\Sigma_n \leq \lambda_1(\Sigma_n)I_n \leq \lambda_1(\Sigma) I_n
\end{equation}
where $I_n$ is the $n \times n$ identity matrix and   $\lambda_1(\Sigma) <\infty$ by assumption, it follows that for all sequences $x, y \in \ell^2$
\begin{equation}
x^n \Sigma_n y^n \leq \lambda_1(\Sigma) \|x^n\|_2\, \|y^n\|_2\,, \qq n=1,2, \ldots
\end{equation}
Then the result follows from the theorem in \cite[p. 53]{Akhiezer-G-61}.

\subsection{Proof of Theorem \ref{th:idiosyncratic}}
\begin{proof}
Assume first that $\lim_{n\rightarrow\infty} \lambda_{1}(\Sigma_n) =
+\infty$. Since $\Sigma_n > 0$ is symmetric  it has a spectral represenattion
\begin{equation}\label{Diagonaliz}
U^\top_n\Sigma_n U_n = D_n \,,
\end{equation}
where $U_n$ is orthonormal and $D_n = \diag \{\,\lambda_{1}(\Sigma_n),\, \ldots ,\, \lambda_{n}(\Sigma_n) \,\}$.
Consider the first column of $U_n$, say $u_1^n$,
which is the eigenvector of $\lambda_{1}(\Sigma_n)$ and define the
sequence of elements in   $\ell^2 \cap \ell^2(\Sigma)$ constructed as
\begin{equation}
a_n := \frac{1}{\sqrt{\lambda_{1}(\Sigma_n)}}\begin{bmatrix} ( u_1^{n}) ^{\top} & 0 & \ldots \end{bmatrix}^\top \,, \qq n=1,2,\ldots\,.
\end{equation}
Since $\lim_{n\rightarrow\infty} \lambda_{1}(\Sigma_n) = +\infty$, this is an AS, for which
\begin{equation}
\|a_n^\top \yb\|^2 =  \frac{1}{ \lambda_{1}(\Sigma_n) } \,(u_1^{n}) ^{\top} \Sigma_n u_1^{n}   = 1
\end{equation}
for every $n$ and hence  the sequence $\yb$ cannot be idiosyncratic.

Conversely, suppose  that $\lambda_1(\Sigma) < +\infty$ and again use
the diagonalization \eqref{Diagonaliz}.  Let $a_n$ be an arbitrary AS and consider the random variable $\zb = \lim_{n \rightarrow \infty} a_n^{\top} \yb =  \lim_{n \rightarrow \infty} a_n^{n\top} \yb^{n}$, which has variance
\begin{equation}
\var[\zb] = \lim_{n\rightarrow\infty} (a_n^{n})^{\top} U_nD_n U_n^{\top}
a^{n}_n := (d_n^{\,n})^{\top} D_n\, d_n^{\,n} \,,
\end{equation}
where the vector $d_n^{\,n} := U_n^{\top} a_n^{n}$ is used to form  the first $n$ elements of  an infinite string, say   $d_n$, whose  remaining entries are taken equal to those of $a_n$; i.e. $d_n(k)= a_n(k)$ for $k >n$. Clearly $d_n$ is an AS.\\
Since $(d_n^{\,n})^{\top} D_n\, d_n^{\,n}= \sum_{k=1}^n \lambda_{k}(\Sigma_n)
d_{n}(k)^2$, one can write
\begin{align*}
\var[\zb] & = \lim_{n\rightarrow\infty} \sum_{i=1}^n \lambda_{k}(\Sigma_n) d_{n}(k)^2 \leq \lim_{n\rightarrow\infty} \lambda_1(\Sigma) \sum_{k=1}^n d_{n}(k)^2 \\ 
&= \lim_{n\rightarrow\infty}\, \lambda_1(\Sigma) \| d_n \|_2^2 = 0
\end{align*}
which shows that $\yb$ is idiosyncratic.
\end{proof}
\subsection{Proof of Theorem \ref{Thm:Strong}}\label{App:C}
\begin{proof}
First we prove the sufficiency of condition \eqref{eq:condition}. Let $k$ be a fixed positive constant and  let
$f_1$ be such that
\begin{equation}
\lim_{n\rightarrow\infty}\| f_1^n  -
\Pi [f_1^n\,|\,\mathcal{F}_1^n] \|_2 = k^{\frac{1}{2}} < +\infty \,.
\end{equation}
Let
\begin{equation}\label{alphas}
\tilde{f}_1^n   = f_1^n - \Pi[f_1^n\,|\,\mathcal{F}_1^n]     = f_1^n - \alpha_2^n f_2^n - \ldots - \alpha_q^n f_q^n \,;
\end{equation}
whence, defining $\tilde{F}^n := \begin{bmatrix} \tilde{f}_1^n
& f_2^n & \ldots & f_q^n \end{bmatrix}$, one can write
$\tilde{F}^n = F^n T^n$,
with $T^n$ is a full rank matrix of the form
\begin{equation} \label{eq:matrixT}
T^n = \begin{bmatrix} 1 & 0 \\ -\alpha_n &I_{q-1} \end{bmatrix} \,,
\end{equation}
where $\alpha_n := \begin{bmatrix} \alpha_2^n & \ldots& \alpha_q^n
\end{bmatrix}^\top$. Since $\tilde{f}_1^n \bot f_i^n $, $i \neq 1$,
  the Gramian matrix of $\tilde{F}^n$ is block diagonal,
\begin{equation}
\tilde{F}^{n\top}\tilde{F}^n = \begin{bmatrix}
\|\tilde{f}_1^n \|^2 & 0 \\ 0 & A_n \end{bmatrix} \,,
\end{equation}
where $A_n$ is a positive definite matrix whose eigenvalues tend to
infinity as $n$ increases. Note that the spectrum of
$\tilde{F}^{n\top}\tilde{F}^n $ contains the eigenvalue
$\|\tilde{f}_1^n \|^2$, which, for $n \rightarrow \infty$,
converges to $k < +\infty$. Now, let us compute the trace of both sides of  the identity $T^n(\tilde{F}^{n\top}\tilde{F}^n)^{-1}T^{n \top}=(F^{n\top}F^n)^{-1}$
 obtaining
\begin{align}
\Tr \left[(F^{n\top}F^n)^{-1}\right] & =\Tr\left[T^n(\tilde{F}^{n\top}\tilde{F}^n)^{-1}T^{n\top}\right] \nonumber \\
 & = \Tr\left[T^{n\top} T^n(\tilde{F}^{n\top}\tilde{F}^n)^{-1}\right]   \nonumber \\
  &  = \Tr \begin{bmatrix}  1 + \|  \alpha_n \|^2 & -\alpha_n^\top  \nonumber\\
-\alpha_n & I_{q-1}  \end{bmatrix} \begin{bmatrix} k^{-1} & 0 \\ 0 &
A_n^{-1} \end{bmatrix} \nonumber\\
& = \Tr \begin{bmatrix}  k^{-1}(1+\|\alpha_n\|^2) & -\alpha_n^\top A_n^{-1} \\ - \alpha_n k^{-1} & A_n^{-1} \end{bmatrix} \nonumber \\
&  =k^{-1}(1+\|\alpha_n\|^2) + \Tr\left[A_n^{-1}\right]
\end{align}
Since the eigenvalues of $A_n$ tend to infinity, those  of
$A_n^{-1}$ tend to zero, while, for every $n$ we have
$k^{-1}(1+\|\alpha_n\|^2) >0$. Thus, one eigenvalue of
$(F^{n\top}F^n)^{-1}$ is bounded below by a fixed constant as $n$ tends to infinity. Hence we
conclude that one eigenvalue of $F^{n\top}F^n$ remains bounded as $
n$ tends to infinity, which is a contradiction.

For the necessity, we define
$f_i^{n_1,n_2}:=\begin{bmatrix}f_i(n_1)&\ldots&f_i(n_2)\end{bmatrix}^\top$
and   observe that condition \eqref{eq:condition} implies that
\begin{equation}
\lim_{n\rightarrow\infty}\| f_i^{n_1,n} -
\Pi[f_i^{n_1,n}\,|\,\mathcal{F}_i^{n_1,n}] \|_2 =  +\infty \,,
\end{equation}
for every index $i = 1,\,\ldots,\,q$ and natural number $n_1$.
Moreover, by definition of limit, we have that for every $n_1 \in
\mathbb{N}$ and $K \in \mathbb{R}_{+}$ there exists an integer $n_2$
such that the inequality (with an obvious   meaning of the symbols)
\begin{equation} \label{eq:ineq}
\| f_i^{n_1,n_2} - \Pi[f_i^{n_1,n_2}\,|\,\mathcal{F}_i^{n_1,n_2}]
\|^2_2 \geq K
\end{equation}
holds for every $i = 1,\,\ldots,\,q$.

Now, consider the sequence generated by the $q$-th eigenvalue of the
matrix $F^{n\top}F^n$, say  $\{\lambda_q^n\,;\,n\in\mathbb{N}\}$.
Our goal is to show that for every natural $n_1$ and arbitrary
constant $c > 0$ there exists a natural number $n_2$ such that
$\lambda_q^{n_2} \geq \lambda_q^{n_1} + c$, so that
$\lim_{n\rightarrow\infty} \lambda_q^n = +\infty$. To   this end, fix
$c$ and, for a generic $n_1$, consider the normalized eigenvector of
the $q$-th eigenvalue of the matrix $F^{n_2\top}F^{n_2}$, say
$v_q^{n_2}$. 
Since for every $n_2 > n_1$ it holds that
\begin{equation}
F^{n_2\top}F^{n_2} = F^{n_1\top}F^{n_1} + F^{n_1,n_2\top}F^{n_1,n_2}
\,,
\end{equation}
 we can write
\begin{equation}
\lambda_q^{n_2} = v_q^{n_2 \top} F^{n_1\top}F^{n_1} v_q^{n_2} +
v_q^{n_2 \top} F^{n_1,n_2\top}F^{n_1,n_2}  v_q^{n_2} \,.
\end{equation}
Consider the first term on the right side of this identity;
expressing $v_q^{n_2}$ as a linear combination of the eigenvectors
of $F^{n_1\top}F^{n_1}$, i.e. $v_q^{n_2} = \alpha_1 v_1^{n_1} +
\ldots + \alpha_q v_q^{n_1}$, the orthogonality of
these eigenvectors implies that
\begin{equation}
v_q^{n_2 \top} F^{n_1\top}F^{n_1} v_q^{n_2}   = \lambda_1^{n_1}
\alpha_1^2 + \ldots + \lambda_q^{n_1} \alpha_q^2
  \geq \lambda_q^{n_1} \sum_{i=1}^q \alpha_i^2 = \!\lambda_q^{n_1}
\end{equation}
so that
\begin{equation} \label{eq:parziale}
\lambda_q^{n_2} \geq \lambda_q^{n_1} + v_q^{n_2 \top}
F^{n_1,n_2\top}F^{n_1,n_2}  v_q^{n_2} \,.
\end{equation}
Now we have to show that we can always find an integer $n_2$ such
that the quantity $$v_q^{n_2 \top} F^{n_1,n_2\top}F^{n_1,n_2}
v_q^{n_2}$$ can be chosen arbitrarily large, i.e. greater or equal to
the previously fixed constant $c$ . To   this end, take $n_2$ such that
for every $i = 1,\,\ldots,\,q$ the inequality \eqref{eq:ineq} holds,
with $K = c\sqrt{q}$. Then, there is an index $i$ such that the $i$-th component of the norm one vector $ v_q^{n_2} =
\begin{bmatrix} w_1 & \ldots & w_q \end{bmatrix}^\top$, satisfies the inequality  $w_i \geq \frac{1}{\sqrt{q}}$. Without loss of generality we may and shall assume that  $i = 1$. Let $\alpha_2 \ldots \alpha_q$ be defined as in \eqref{alphas} and set
\begin{equation} \label{eq:cambio_base}
\tilde{f}_1^{n_1,n_2} := f_1^{n_1,n_2} - \alpha_2 f_2^{n_1,n_2} -
\ldots - \alpha_q f_q^{n_1,n_2} \,,
\end{equation}
so that we have
\begin{equation} \label{eq:prod}
v_q^{n_2 \top} F^{n_1,n_2\top}F^{n_1,n_2} v_q^{n_2} = v_q^{n_2 \top}
T^{n \top}\!\! \begin{bmatrix} \|\tilde{f}_1^{n_1,n_2} \|^2 \!&\! 0\! \\ \!0 \!&\! A_n \!\end{bmatrix} T^n v_q^{n_2}
\end{equation}
where $T^n$ has the same structure as in \eqref{eq:matrixT}. Now,
observe that
\begin{equation}
T^n v_q^{n_2} = \begin{bmatrix} w_1 & -\alpha_2 w_1 + w_2 & \ldots &
-\alpha_q w_1 + w_q \end{bmatrix}^\top \,,
\end{equation}
which implies  that \eqref{eq:prod} is equal to $w_1^2
\|\tilde{f}_1^{n_1,n_2} \|^2 + Q ,$ where $Q$ is a positive
constant. Hence, from \eqref{eq:cambio_base} we have $v_q^{n_2
\top} F^{n_1,n_2\top}F^{n_1,n_2}  v_q^{n_2} > c$  and hence, recalling
\eqref{eq:parziale},
\begin{equation}
\lambda_q^{n_2} \geq \lambda_q^{n_1} + c \,.
\end{equation}
which proves the theorem.
\end{proof}

\subsection{Proof of Proposition \ref{DynStationary}}\label{App:E}
\begin{proof}
For infinite covariance matrices we  have the  positive semidefinite ordering $\Sigma_1 \leq \Sigma_2$ if and only if $a^{\top} (\Sigma_1-\Sigma_2)a \leq 0$ for all finite support  sequences $a \in \Rbb^{\infty}$. Let $\yb(t)$ be a time-stationary random field   with matrix covariance function $\Sigma(\tau):= \E\yb(t+\tau)\yb(t)^{\top}$. For any finite support sequence $a\in \Rbb^{\infty}$  the scalar covariance function $\sigma_{\zb}(\tau)$ of the process $\zb(t):=a^{\top}\yb(t)$ satisfies the  well-known (Schwartz) inequality $\sigma_{\zb}(\tau) \leq \sigma_{\zb}(0)$; hence the matrix covariance function  of a  stationary process satisfies  $\Sigma(\tau) \leq \Sigma(0)$. It follows that if $\Sigma(0)$ is a bounded operator in $\ell^2$ then all covariances $\Sigma(\tau)$ must also be bounded. The following  lemma is a straightforward consequence of this fact.
\begin{lemma} \label{PropIdio}
A time-stationary random field   $\yb(t)$ is idiosyncratic; that is 
$$
\lim_{n\to \infty} \,a_n^{\top}\yb(t)=0 \,, \q \text{ for \;\;any} \q t\in \Zbb
$$
for all AS's $a_n$,  if and and only if $\yb(0)$ is an idiosyncratic sequence.
 \end{lemma}
 The lemma above implies in particular that   a covariance function   $\Sigma(\tau)$ is the covariance of an idiosyncratic stationary random field iff  $ \Sigma(0)$ is a   bounded   operator on $\ell^2$.

Now assume that $\Sigma(0)$  has a static GFA decomposition of rank $q$ and let $\xb$ and $\tilde{\yb}$ be constructed as in the proof of Proposition \ref{Prop:Constr} so that the vector $\yb(0)$ has a GFA representation
$$
\yb(0)= F \xb +\tilde{\yb}\,
$$
where $\xb$ and $\tilde{\yb}$ have uncorrelated  components belonging to $H(\yb(0))$. Let $\Hb(\yb)$ denote the closed linear span of the scalar components of the random field $\yb$; i.e.
$$
\Hb(\yb):= \text{closure of} \left\{ \sum_{k,t} a_{k,t} \yb(k,t)\,;\, k=0,1,2,\ldots; t \in \Zbb \right\}
$$
where the    real numbers $ a_{k,t}$ are arbitrary  but non zero only for a finite set of values of the indices. Let
$U: \Hb(\yb)\rightarrow \Hb(\yb)$ be the forward shift operator of the process defined, for all finite support vectors $a$, by
$$
U a^{\top}\yb(t) = a^{\top}\yb(t+1)\,, \qquad t\in \Zbb
$$
It is well known that $U$ can be extended to the whole of $\Hb(\yb)$ as  a unitary operator \cite{Rozanov-67} and that every scalar random variable $ \zb \in \Hb(\yb)$ can be propagated in time by the action of the shift as $\zb(t):= U^{t} \zb$ to form a   stationary scalar process. This unitary propagation can in fact be applied to vector random variables of arbitrary dimension.\\
It follows hat  $\xb(t):= U^t \xb$ and  $\tilde{\yb}(t):=U^t \tilde{\yb}$ have uncorrelated components for all $t$. Moreover, by Lemma \ref{PropIdio} the stationary process $\tilde{\yb}(t)$ is idiosyncratic and  $\hat{\yb}(t):= F\xb(t)$ is  a flocking process since the columns of $F$ are strongly linearly independent. 
 \end{proof}

\subsection{Proof of Proposition  \ref{thm:leader}} \label{App:F}
Consider first the case $n=3$ and write the solution to the related Lyapunov equation as
\begin{equation}
\Sigma_3 = \begin{bmatrix} p_1 & p_2 & p_3 \\ p_2 & p_4 & p_5 \\ p_3 & p_5 & p_6 \end{bmatrix} \,.
\end{equation}
Then, simple calculations show that
$$
p_1 = \frac{1}{1-a^2} \quad,\quad p_2 = p_3 = \frac{a}{(1+a)(1-a^2)} \quad,
$$
\begin{align}
p_4 &= p_6 = \frac{1}{1-a^2}+\frac{1}{(1+a)^2}+2\frac{a^2}{(1+a)^2(1-a^2)} \nonumber \\
p_5 &= \frac{1}{(1+a)^2} +2\frac{a^2}{(1+a)^2(1-a^2)}
\end{align}
Now assume that, for a given $n \geq 3$, the solution to the equation $X_n - A_nX_nA_n^\top = I_n$ has the form
\begin{equation} \label{eq:sol_n}
\Sigma_n = \begin{bmatrix} p_1  & p_3 & p_3 & p_3& \ldots & p_3 \\
                           p_3  & p_4 & p_5 & p_5 & \ldots & p_5 \\
                           p_3 & p_5 &  p_4 & p_5 &  \ldots & p_5 \\
                           \vdots & \vdots & \ddots & \ddots & \ddots & \vdots \\
                           p_3 & p_5 & \ldots & p_5 & p_4 & p_5 \\
                           p_3 & p_5 & \ldots & p_5 & p_5 & p_4 \end{bmatrix} \,;
\end{equation}
our goal is to show that $\Sigma_{n+1}$ has an analogous structure, that is
\begin{equation} \label{eq:sol_n+1}
\Sigma_{n+1} = \begin{bmatrix} \Sigma_n & p \\ p^\top & p_4 \end{bmatrix} \,,
\end{equation}
where $p = \begin{bmatrix} p_3 & p_5 & \ldots &p_5  \end{bmatrix}^\top$. To this end, express the variable $X_{n+1}$ as
$$
X_{n+1} = \begin{bmatrix} X_n & z \\ z^\top & u \end{bmatrix} \,
$$
and the matrix $A_{n+1}$ as
$$
A_{n+1} = \begin{bmatrix} A_n & 0 \\ b^\top & a \end{bmatrix} \,,
$$
where $b = \begin{bmatrix} 1-a & 0 & \ldots &0  \end{bmatrix}^\top$. Then the related Lyapunov equation has the form
\begin{equation}
\begin{bmatrix} X_n & z \\ z^\top & u \end{bmatrix} - \begin{bmatrix} A_n & 0 \\ b^\top & a \end{bmatrix} \begin{bmatrix} X_n & z \\ z^\top & u \end{bmatrix} \begin{bmatrix} A_n^\top & b \\ 0 & a \end{bmatrix} = I_{n+1} \,,
\end{equation}
which can be rewritten as
\begin{align} \label{eq:lyapn+1}
&\begin{bmatrix} X_n - A_n X_n A_n^\top & (I_n-aA_n)z - A_n X_n b \\ z^\top(I_n\!-\!aA_n^\top) -   b^\top X_n A_n^\top & (1\!-\! a^2)u -  b^\top\! X_n b - 2a b^\top\! z\end{bmatrix} \!\!= \nonumber\\ 
&=\begin{bmatrix}I_n & 0 \\ 0 & 1 \end{bmatrix} \,.
\end{align}
The top-left block of \eqref{eq:lyapn+1} admits the solution given by \eqref{eq:sol_n}. Then, by inserting this into the top-right block, one then gets $z = p$. Finally, by  exploiting the former findings, from the bottom-right block one has $u = p_4$, and hence the solution is exactly \eqref{eq:sol_n+1}. Hence, one can easily observe that the matrix $\bar \Sigma_n$, obtained by discarding the first row and column from $\Sigma_n$, has the structure
\begin{equation}
  \begin{bmatrix}  p_5 & p_5 & \ldots \\ p_5 & p_5 & \\ \vdots & & \ddots \end{bmatrix} + \diag \{p_4 - p_5,\,\ldots,\, p_4 - p_5\} \,
\end{equation}
that is, it admits a rank-one plus diagonal decomposition, where the vector generating the rank-one matrix is $\bar f =  \begin{bmatrix} \sqrt{p_5} &  \sqrt{p_5} & \ldots  \end{bmatrix}$, with $\sqrt{p_5} = (1+a^2)^{\frac{1}{2}}/((1+a)(1-a^2)^{\frac{1}{2}})$, while the elements of the diagonal matrix are $p_4 - p_5 = 1/(1-a^2)$. Now, to complete the proof we need to show that also the matrix $\Sigma_n$ admits a similar decomposition, i.e.  
$$
\Sigma_n = \begin{bmatrix} f_0 \\ \bar f \end{bmatrix} + \diag \{\sigma_0^2,\, 1/(1-a^2),\,\ldots,\, 1/(1-a^2)\} \,.
$$
This can be done be observing that, for any integer $k >0$, it has to be $p_3 = f_0 \bar f(k)$, and so $f_0 = a/(1-a^4)^{\frac{1}{2}}$. Moreover, $\sigma_0^2$ is easily found by computing $\sigma_0^2 = p_1 - f_0^2 = 1$. Finally, since by comparing the leader dynamics
$$
\yb_0(t) = a \yb_0(t-1) + \wb_0(t-1)
$$
with its GFA decomposition
$$\yb_0(t) = f_0 \xb(t) + \tilde \yb_0 (t) \,,$$
where both $\tilde \yb_0 (t)$ and $\wb_0(t-1)$ are white noise with the same variance, it has to be that $\xb(t) = (1-a^4)^{\frac{1}{2}} \yb_0(t-1)$.

\subsection{Proof of Proposition \ref{thm:consensus}}\label{App:G}
For   $n \geq n_0$, consider the Lyapunov equation
$$ \Sigma_n = \bar A_n \Sigma_n \bar A_n^\top + I_n \,,$$
whose solution can be written
\begin{equation}
\Sigma_n = \sum_{j=0}^\infty \bar A_n^j (\bar A_n^j)^{\top} \,.
\end{equation}
Since $\bar A_n$ is symmetric, for every $j$ the decomposition
$$\bar A_n^j (\bar A_n^j)^T = U_n S_n^{2j} U_n^{\top}$$
holds, with $S_n$ being the matrix of the singular values of $A$ and $U_n$ a unitary matrix whose columns are the (normalized) eigenvectors of $\bar A_n$.
Note that one of such singular values is $\left(1 - \frac{1}{n} \right)^{2}$  and the relative eigenvector is $\frac{1}{\sqrt n}{\mathbf 1}\!\!1_n$, i.e. the normalized vector of all  $1$'s in $\in \Rbb^n $. The other eigenvalues are strictly stable. Then we can express $\Sigma_n$ as
\begin{align} \label{eq:proof_consensus}
\Sigma_n & = U_n \left( \sum_{j=0}^\infty S_n^{2j} \right) U_n^\top \nonumber \\
         & = \frac{{\mathbf 1}\!\!1}{\sqrt n} \left( \sum_{j=0}^\infty \left(1 - \frac{1}{n} \right)^{2j} \right) \frac{{\mathbf 1}\!\!1}{{\sqrt n}}^\top + \tilde U_n \left( \sum_{j=0}^\infty \tilde S_n^{2j} \right) \tilde U_n^\top \nonumber \\
         & = {\mathbf 1}\!\!1 \frac{n}{2n+1} {\mathbf 1}\!\!1^\top + \tilde U_n \left( \sum_{j=0}^\infty \tilde S_n^{2j} \right) \tilde U_n^\top \,,
\end{align}
where $\tilde U_n$ and $\tilde S_n$ are obtained from $U_n$ and $S_n$ by removing the parts related to the eigenvalue $\left(1 - \frac{1}{n} \right)^{2}$. Now,   take the averaging sequence \eqref{Aritmean}
\begin{equation} \label{eq:as_cons}
a_n = \frac{1}{n}\left[{\mathbf 1}\!\!1_n^\top \, 0 \, \ldots \right] \quad , \, {\mathbf 1}\!\!1_n \in \mathbb{R}^n
\end{equation}
and apply it to $\Sigma_n$, that is,   compute $\frac{1}{n}{\mathbf 1}\!\!1^\top_n \Sigma_n {\mathbf 1}\!\!1_n\frac{1}{n} $.  Then, letting $n \rightarrow \infty$, the second term on the right hand side of \eqref{eq:proof_consensus} vanishes, while the first term gives
\begin{equation}
\frac{{\mathbf 1}\!\!1^\top_n  {\mathbf 1}\!\!1_n {\mathbf 1}\!\!1^\top_n  {\mathbf 1}\!\!1_n}{n(2n+1)} = \frac{n}{2n+1}\,,
\end{equation}
which converges asymptotically to a finite value. One can easily verify that the averaging sequence \eqref{eq:as_cons} is the only sequence converging to nonzero values.

\bibliographystyle{plain}

\bibliography{generalized_factor_models}

\end{document}